%% file: main.tex
\documentclass[journal]{IEEEtran}
\usepackage[utf8]{inputenc} % allow utf-8 input
\usepackage[T1]{fontenc}    % use 8-bit T1 fonts
\usepackage{hyperref}       % hyperlinks
\usepackage{url}            % simple URL typesetting
\usepackage{booktabs}       % professional-quality tables
\usepackage{microtype}      % microtypography
\usepackage{amsfonts}       % blackboard math symbols
\usepackage{nicefrac}       % compact symbols for 1/2, etc.

\usepackage{bm}
\usepackage{amsmath,amssymb,amsthm}
\usepackage{mathtools}

\usepackage{multicol}
\usepackage{wrapfig}

\usepackage{algorithm}
\usepackage{algpseudocode}

\usepackage{tikz}
\usetikzlibrary{arrows}
\usetikzlibrary{calc}
\usepackage{pgfplots}
\pgfplotsset{compat=newest}
\pgfplotsset{plot coordinates/math parser=false}
\newlength\figureheight
\newlength\figurewidth
\usepackage{dashrule}

\usepackage{subcaption}
\input{definitions.tex}

\newcounter{steps}

\newcommand{\step}[1][\empty]% #1 = label (optional)
{\stepcounter{steps}%
 \par
 \hangindent=4em
 \hangafter=1
 \makebox[4em][l]{\textit{Step \arabic{steps}:}}%
 \ifx#1\empty\else #1 --\fi
}
\newenvironment{labelledsteps}%
{\setcounter{steps}{0}%
 \let\item=\step
 \parskip=0.25\baselineskip
 \parindent=0pt}%
{\par}

\begin{document}

\title{Probabilistic model predictive safety certification for learning-based control}

\author{Kim P.~Wabersich,
        Lukas~Hewing,
        Andrea~Carron, 
        Melanie~N.~Zeilinger % <-this % stops a space
\thanks{The authors are members of the Institute for Dynamic Systems and Control, ETH Z\"urich, Z\"urich
CH-8092, Switzerland (e-mail: [wkim|lhewing|carrona|mzeilinger]@ethz.ch)}% <-this % stops a space
\thanks{This work was supported by the Swiss National Science Foundation under grant no. PP00P2 157601/1. Andrea Carron's research was supported by the Swiss National Centre of Competence
in Research NCCR Digital Fabrication (Agreement \#51NF40\_141853).}}

% note the % following the last \IEEEmembership and also \thanks - 
% these prevent an unwanted space from occurring between the last author name
% and the end of the author line. i.e., if you had this:
% 
% \author{....lastname \thanks{...} \thanks{...} }
%                     ^------------^------------^----Do not want these spaces!
%
% a space would be appended to the last name and could cause every name on that
% line to be shifted left slightly. This is one of those "LaTeX things". For
% instance, "\textbf{A} \textbf{B}" will typeset as "A B" not "AB". To get
% "AB" then you have to do: "\textbf{A}\textbf{B}"
% \thanks is no different in this regard, so shield the last } of each \thanks
% that ends a line with a % and do not let a space in before the next \thanks.
% Spaces after \IEEEmembership other than the last one are OK (and needed) as
% you are supposed to have spaces between the names. For what it is worth,
% this is a minor point as most people would not even notice if the said evil
% space somehow managed to creep in.

% The paper headers
\markboth{}%
{Wabersich \MakeLowercase{\textit{et al.}}}

% The only time the second header will appear is for the odd numbered pages
% after the title page when using the twoside option.
% 
% *** Note that you probably will NOT want to include the author's ***
% *** name in the headers of peer review papers.                   ***
% You can use \ifCLASSOPTIONpeerreview for conditional compilation here if
% you desire.

% If you want to put a publisher's ID mark on the page you can do it like
% this:
%\IEEEpubid{0000--0000/00\$00.00~\copyright~2015 IEEE}
% Remember, if you use this you must call \IEEEpubidadjcol in the second
% column for its text to clear the IEEEpubid mark.

% use for special paper notices
%\IEEEspecialpapernotice{(Invited Paper)}

% make the title area
\maketitle

\begin{abstract}
Reinforcement learning (RL) methods have demonstrated their efficiency in simulation. However, many of the applications for which RL offers great potential, such as autonomous driving, are also safety critical and require a certified closed-loop behavior in order to meet safety specifications in the presence of physical constraints. This paper introduces a concept called probabilistic model predictive safety certification (PMPSC), which can be combined with any RL algorithm and provides provable safety certificates in terms of state and input chance constraints for potentially large-scale systems. The certificate is realized through a stochastic tube that safely connects the current system state with a terminal set of states that is known to be safe. A novel formulation allows a recursively feasible real-time computation of such probabilistic tubes, despite the presence of possibly unbounded disturbances. A design procedure for PMPSC relying on Bayesian inference and recent advances in probabilistic set invariance is presented. Using a numerical car simulation, the method and its design procedure are illustrated by enhancing an RL algorithm with safety certificates.
\end{abstract}


\IEEEpeerreviewmaketitle

\section{Introduction}
\IEEEPARstart{W}{hile} the field of reinforcement learning demonstrated various classes
of learning-based control methods in research-driven applications~\cite{mnih2015atari,merel2017humanBehavior},
% which have been demonstrated in research-driven applications, e.g.~\cite{merel2017humanBehavior},
very few results have been successfully transferred to industrial applications that are
\textit{safety-critical},
i.e.\ applications that are subject to physical and safety constraints.
In industrial applications, successful control methods are often
of simple structure, such as the Proportional–Integral–Derivative (PID) controller
\cite{ang2005pid} or linear state feedback controller~\cite{skogestad2007multivariable},
which require an expert to cautiously tune them manually.
Manual tuning is generally time consuming and therefore expensive, especially in the presence
of safety specifications. Modern control methods, such as model predictive control (MPC),
tackle this problem by providing safety guarantees with respect to adequate system and disturbance models by
design, reducing manual tuning requirements. The various successful applications of MPC
to safety critical systems reflect these capabilities, see e.g.~\cite{qin2000overview,lee2011model} for
an overview.

While provable safety of control methods facilitates the overall design procedure, the
tuning of various parameters, such~as the cost function, in order to achieve
a desired closed-loop behavior, still needs to be done manually and often requires
significant experience. In contrast, RL methods using trial-and-error procedures are often more
intuitive to design and are capable of iteratively computing an improved policy.
The downside of many RL algorithms, however, is that explicit consideration of physical system
limitations and safety requirements at each time step cannot be addressed, often due to
the complicated inner workings, and this limits their applicability in many industrial
applications~\cite{amodei2016concrete}.

This paper aims to address this problem by introducing a probabilistic model predictive safety
certification (PMPSC) scheme for learning-based controllers, which can equip any controller with
probabilistic constraint satisfaction guarantees. The scheme is motivated by the following
observation. Often, an MPC controller with a short prediction horizon is sufficient in order to provide safety
for a system during a closed-loop operation, even though the same horizon would not be enough
to achieve a desired performance. For example, in the case of
autonomous driving, checking if it is possible to transition the car into a safe set of states
(e.g.\ brake down to low velocity) can be done efficiently by solving an open loop optimal
control problem with a relatively small planning horizon (e.g.\ using maximum deceleration).
At the same time, a much longer planning horizon for an MPC controller, or even another
class of control policies, would be required in order to provide a comfortable
and foresightful driving experience.

This motivates the combination of ideas from MPC with RL methods in order to achieve
a safe and high performance closed-loop system operation requiring a small
amount of manual tuning. More precisely, a learning-based input action is
certified as safe if it leads to a safe state, i.e., a state for which a
potentially low-performance, but online computable and safe backup controller
exists for all future times.
By repeatedly computing such a backup controller for the state predicted one step ahead
after application of the learning input, it is either certified
as safe and is applied, or it is overwritten by the previous safe backup controller.
The resulting concept can be seen as a safety filter
that only filters proposed learning signals for which we cannot
guarantee constraint satisfaction in the future.

\textit{Contributions:}
We provide a safety certification framework which allows for enhanced
arbitrary learning-based control methods with safety guarantees\footnote{Inputs
provided by a human can be similarly enhanced by the safety certification scheme,
which relates e.g.\ to the concept of electronic stabilization control from
automotive engineering.} and which is suitable for possibly large-scale
systems with continuous and chance constrained input and state spaces. In order to
enable efficient implementation and scalability,
we provide an online algorithm together with a data-driven synthesis method
to compute backup solutions that can be realized by real-time capable
and established model predictive control (MPC) solvers,
e.g.~\cite{wang2010fastMpc,domahidi2012efficientInteriorPoint,houska2011acado}.
Compared to previously presented safety frameworks for learning-based control, e.g.
\cite{fisac2017generalSafetyFramework}, the set of safe state and action
pairs is implicitly represented through an online optimization problem, enabling us
to circumvent its explicit offline computation, which generally suffers from
the curse of dimensionality.

Unlike related concepts, such as those presented in
\cite{wabersich2018linear,Gurriet2018,Mannucci2018,Bastani2019},
we consider possibly nonlinear stochastic systems that can be represented
as linear systems with bounded model uncertainties and possibly unbounded
additive noise. For this class of systems, we present an automated, parametrization
free and data-driven design procedure that is tailored to the context of learning
the system dynamics. Due to our specific formulation, we can maintain recursive
feasibility of the underlying optimization problem, which is a distinctive difference
to related stochastic MPC and safety filter schemes. Using the example of safely
learning to track a trajectory with a car, we show how to construct a safe
reinforcement learning algorithm using our framework in combination with a basic
policy search algorithm.

\section{Related Work}
Driven by rapid progress in reinforcement learning there is also a
growing awareness regarding safety aspects of machine learning systems
\cite{amodei2016concrete}, see e.g.~\cite{garcia2015aComprehensiveSurveySafeReinforcementLearning}
for a comprehensive overview. As opposed to most methods
developed in the context of safe RL, the approach presented in this paper
keeps the system safe at all times, including exploration, and considers
continuous state and action spaces. This is possible through the use of
models and corresponding uncertainty estimates of the system, which can
be sequentially improved by, e.g., an RL algorithm to allow greater exploration.

In model-free safe reinforcement learning methods, policy search algorithms
have been proposed, e.g.~\cite{achiam2017constrained}, which provide expected safety
guarantees by solving a constrained policy optimization using
a modified trust-region policy gradient method~\cite{schulman2015trust}.
Efficient policy tuning with respect to best worst-case performance (also worst-case
stability under physical constraints) can be achieved using Bayesian min-max
optimization, see e.g.~\cite{wabersich2015automatic}, or by
safety-constrained Bayesian optimization as e.g.\ in
\cite{berkenkamp2015SafeRobustLearning,schreiter2015safe}.
These techniques share the limitation that they need to be tailored to a task-specific
class of policies. Furthermore, most techniques require repeated
execution of experiments, which prohibits fully autonomous safe learning in
`closed-loop'.

In~\cite{berkenkamp2016safe}, a method was developed that
allows for the analysis of a given closed-loop system (under an arbitrary RL policy)
with respect to safety, based on a probabilistic system model. An extension of this method is
presented in~\cite{berkenmap2017safe}, where the problem of updating the policy is
investigated and practical implementation techniques are provided.
The techniques require an a-priori known Lyapunov function and Lipschitz
continuity of the \textit{closed-loop} learning system. In the context of model-based safe
reinforcement learning, several learning-based model predictive control approaches are available.
The method proposed in~\cite{aswani2013safe} conceptually provides deterministic guarantees on
robustness, while statistical identification tools are used to identify the system in order to improve performance.
In \cite{bouffard2012learning}, the scheme mentioned has been tested and validated onboard
using a quadcopter. In \cite{Ostafew2016robust}, a robust constrained learning-based model predictive
control algorithm for path-tracking in off-road terrain is studied. The experimental evaluation shows that
the scheme is safe and conservative during initial trials, when model uncertainty is high and very
performant once the model uncertainty is reduced. Regarding safety, \cite{Koller2018} presents a
learning model predictive control method that provides theoretical guarantees in the case of Gaussian
process model estimates. For iterative tasks, \cite{Rosolia2017a} proposes a learning model predictive
control scheme that can be applied to linear system models with bounded disturbances. Instead of using
model predictive control techniques, PILCO~\cite{Deisenroth2011} allows the calculation of analytic policy gradients
and achieves good data efficiency, based on non-parametric Gaussian process regression.

The previously discussed literature provides specific reinforcement learning algorithms that are tied
to a specific, mostly model predictive control based, policy. In contrast, the proposed concept
uses MPC-based ideas in order to establish safety independently of a specific reinforcement
learning policy. This offers the opportunity to apply RL for learning more complex tasks than
for example steady-state stabilization, which is usually considered in model predictive control.
Many reinforcement learning algorithms are able to maximize rewards from
a black-box function, i.e.\ rewards that are only available through measurements, which would
not be possible using a model predictive controller, where the cost enters the corresponding
online optimization problem explicitly.

Closely related to the approach proposed in this paper, the concept of
a safety framework for learning-based control emerged from robust reachability analysis,
robust invariance, as well as classical Lyapunov-based methods~\cite{Gillula2011,
fisac2017generalSafetyFramework,Wabersich2018c,larsen2017safeLearningDistributed}.
The concept consists of a safe set in the state space and a safety controller
as originally proposed in \cite{Wieland2007} for the case of perfectly known
system dynamics in the context of safety barrier functions.
While the system state is contained in the safe set, any feasible input
(including learning-based controllers) can be applied
to the system. However, if such an input would cause the system to leave the
safe set, the safety controller interferes. Since this strategy is compatible
with any learning-based control algorithm, it serves as a universal
safety certification concept. Previously proposed concepts are limited
to a robust treatment of the uncertainty in order to provide rigorous
safety guarantees. This potentially results in a conservative
system behavior, or even the ill-posedness of the overall safety requirement
e.g.\ in the case of frequently considered Gaussian distributed additive
system noise, which has unbounded support.

Compared to previous research using similar model predictive control-based
safety mechanisms such as \cite{wabersich2018linear,Gurriet2018,Mannucci2018,Bastani2019}, we
introduce a probabilistic formulation of the safe set and consider
safety in probability for all future times, allowing one to prescribe a
desired degree of conservatism and address disturbance distributions 
with unbounded support. The proposed method only requires an
implicit description of the safe set as opposed to an explicit
representation, which enables scalability with respect to the
state dimension, while being independent of a particular RL algorithm.

\section{Preliminaries and Problem Statement}\label{sec:preliminaries_and_problem}

\subsection{Notation}
The set of symmetric matrices of dimension $n$
is denoted by $\mSetSymMat{n}$, the set of positive (semi-)
definite matrices by ($\mSetPosSemSymMat{n}$) $\mSetPosSymMat{n}$,
the set of integers in the interval $[a,b]\subset\RR$ by
$\mIntInt{a}{b}$, and the set of integers in the interval
$[a,\infty)\subset\RR$ by $\mIntGeq{a}$.
The Minkowski sum of two sets $\mathcal A_1, \mathcal A_2 \subset \RR^n$
is denoted by
$\mathcal A_1 \oplus \mathcal A_2 \mDef \{ a_1 + a_2 | a_1 \in \mathcal A_1, a_2 \in \mathcal A_2 \}$
and the Pontryagin set difference by
$\mathcal A_1 \ominus \mathcal A_2 \mDef \{ a_1 \in \RR^n | a_1 + a_2 \in \mathcal A_2,
~ \forall a_2 \in \mathcal A_2 \} $. An affine image of a set
$\mathcal A_1\subseteq \RR^n$ under $x\mapsto K x$ is defined as $K\mathcal A_1\mDef\{K x|x\in \mathcal A_1\}$.
The $i$-th row and $i$-th column of
a matrix $A\in\RR^{n\times m}$ are denoted by $\mRow{i}(A)$
and $\mCol{i}(A)$, respectively. The expression $x\sim \mDistribution{x}$ means that a random variable $x$
is distributed according to the distribution $\mDistribution{x}$, and
$\NN(\mu,\Sigma)$ is a multivariate Gaussian distribution with mean $\mu\in\RR^n$
and covariance $\Sigma \in \mSetPosSemSymMat{n}$. The probability
of an event $E$ is denoted by $\Pr(E)$. For a random variable $x$, $\mExpectation{x}$
and $\mVariance{x}$  denote the expected value and the variance.

\subsection{Problem statement}
We consider a-priori unknown nonlinear, time-invariant discrete-time dynamical
systems of the form
\begin{align}\label{eq:general_nonlinear_system}
  x(k\!+\!1) = f_{\theta}(x(k), u(k)) +  w_s(k), ~ \forall k\in\mIntGeq{0}
\end{align}
subject to polytopic state and input constraints $x(k) \in \XX$, $u(k) \in \UU$,
and i.i.d.\ stochastic disturbances $w_s(k) \sim \mDistribution{w_s}$. 
The uncertainties in the function $f_\theta$ are characterized by the random 
parameter vector $\theta \sim \mDistribution{\theta}$.
For controller design, we consider an approximate model description
of the following form
\begin{align}\label{eq:general_linear_additive_description}
    x(k+1) = Ax(k) + Bu(k) + w_{\theta}(x(k),u(k)) + w_s(k),
\end{align}
where $A\in\RR^{n\times n}$, $B\in\RR^{n\times m}$ are typically obtained
from linear system identification techniques, see e.g.\ \cite{ljung1998system},
and $w_{\theta}(x(k),u(k))$ accounts for model errors.

In order to provide safety certificates, we require that the model error $w_{\theta}(x(k),u(k))$
is contained with a certain probability in a model error set $\mAddDisturbance_{\theta}\subset\RR^n$, 
where $\mAddDisturbance_{\theta}$ is chosen based on available data
$\mData = \{(x_i,u_i,f_{\theta}(x_i,u_i) + w_{s,i})\}_{i=1}^{N_\mData}$,
where $N_\mData$ denotes the number of available data points.
\begin{assumption}[Bounded model error]\label{ass:bound_on_nonlinearities}
    The deviation between the true system~\eqref{eq:general_nonlinear_system}
    and the corresponding model~\eqref{eq:general_linear_additive_description} is
    bounded, i.e.
    \begin{align}
        \Pr
        \begin{pmatrix}
            w_{\theta}(x(k),u(k)) \in \mAddDisturbance_{\theta} \\
            \forall k\in\mIntGeq{0},x(k)\in\RR^n,u(k)\in\RR^m
        \end{pmatrix}
        \geq p_{\theta}
        \label{eq:bound_on_nonlinearities}
    \end{align}
    where $p_{\theta}>0$ denotes the probability level 
    and the probability is taken with respect to\ the random parameter $\theta$. \END
\end{assumption}
A principled way to infer a system model of the form~\eqref{eq:general_linear_additive_description} for linear systems from available data 
such that Assumption~\ref{ass:bound_on_nonlinearities} is satisfied for a compact domain is discussed in Section \eqref{sec:data_driven_design}.
For nonlinear systems, the computation of $\mAddDisturbance_\theta$ typically involves the solution 
of a non-convex optimization problem, which can be approximated e.g.\ by gridding, similarly as proposed
in \cite[Remark IV.1]{Wabersich2018c}.
\begin{remark}[Gaussian processes]\label{rem:gaussian process}
    Instead of parametric uncertainty, one could also use
	non-parametric Gaussian process regression for the dynamics function $f_\theta$.
    The model error set $\mAddDisturbance_{\theta}$ can then be derived
    by assuming that $f_\theta$ has a bounded norm in a reproducing kernel Hilbert space using the bound presented in \cite[Theorem~2]{Chowdhury2017}. For function samples of a Gaussian process on compact domains, one can similarly apply, e.g., the bound presented in~\cite{Lederer2019GPbound}.
	\END
\end{remark}

In this paper, system safety is defined as a required degree of constraint
satisfaction in the form of probabilistic state and input constraints, i.e.,
as chance-constraints of the form
\begin{align}\label{eq:chance_constraints}
	\Pr(x(k)\in\XX) \geq p_x\, , \ \Pr(u(k)\in\UU) \geq p_u\, ,
\end{align}
for all $k\in\mIntGeq{0}$ with probabilities $p_x,p_u \geq 0$, where the probability 
is with respect to\ all uncertain elements, i.e.\ parameters $\theta$ and noise realizations $\{w_s(i)\}_{i=0}^{k-1}$.

The overall goal is to certify safety of arbitrary control signals $u_\LL(k)\in\RR^m$,
e.g, provided by an RL algorithm. This is achieved by means of a
\emph{safety policy}, which is computed in real time based on the current system
state $x(k)$ and the proposed input $u_\LL(k)$.
A safety policy consists of a safe input $u_\mSafeSet(k)$ at time $k$
and a safe backup trajectory that guarantees safety with respect
to the constraints~\eqref{eq:chance_constraints} when applied in future time
instances. The safety policy is updated at every time step, such that the
first input equals $u_\LL(k)$ if that is safe and otherwise implements a
minimal safe modification. More formally:
\begin{definition}\label{def:safety_certification}
    Consider any given control signal $u_\mathcal{L}(k)\in\RR^m$ for
    time steps $k\in\mIntGeq{0}$. We call a control input $u_\LL(\bar k)$~
    \emph{certified as safe} for system~\eqref{eq:general_nonlinear_system} at
    time step $\bar{k}$ and state $x(\bar{k})$ with respect to a \emph{safety policy}
    $\mDefFunction{\pi_\mSafeSet}{\RR^n\times\RR^m}{\RR^m}$,
    if $\pi_\mSafeSet(x(\bar k), u_\LL(\bar k)) = u_\LL(\bar k)$ and $u(k) = \pi_\mSafeSet(x(k), u_\LL(k))$
    keeps the system safe, i.e. \eqref{eq:chance_constraints}
    is satisfied for all $k\geq 0$.
    \END
\end{definition}

By assuming that a safety policy can be found for the initial system state,
Definition~\ref{def:safety_certification} implies the following safety algorithm. At every
time step, safety of a proposed input $u_\LL(k)$ is verified using the safety policy
according to Definition~\ref{def:safety_certification}. If safety cannot be verified,
the proposed input is modified and
$u(\bar{k}) = \pi_\mSafeSet(x(\bar{k}),u_\LL(\bar k))$ is applied to the
system instead, ensuring safety until the next state and learning input pair
can be certified as safe again. The set of initial states for which $\pi_\mSafeSet$
ensures safety can thus be interpreted as a \emph{safe set} of system states
and represents a probabilistic variant of the safe set definition in
\cite{wabersich2018linear}.

In the following, we present a method to compute a safety policy
$\pi_\mSafeSet$ for uncertain models of the form~\eqref{eq:general_linear_additive_description}
making use of model predictive control (MPC) techniques, which provide real-time
feasibility and scalability of the approach while aiming at a
large safe set implicitly defined by the safety policy.

\section{Probabilistic model predictive safety certification}
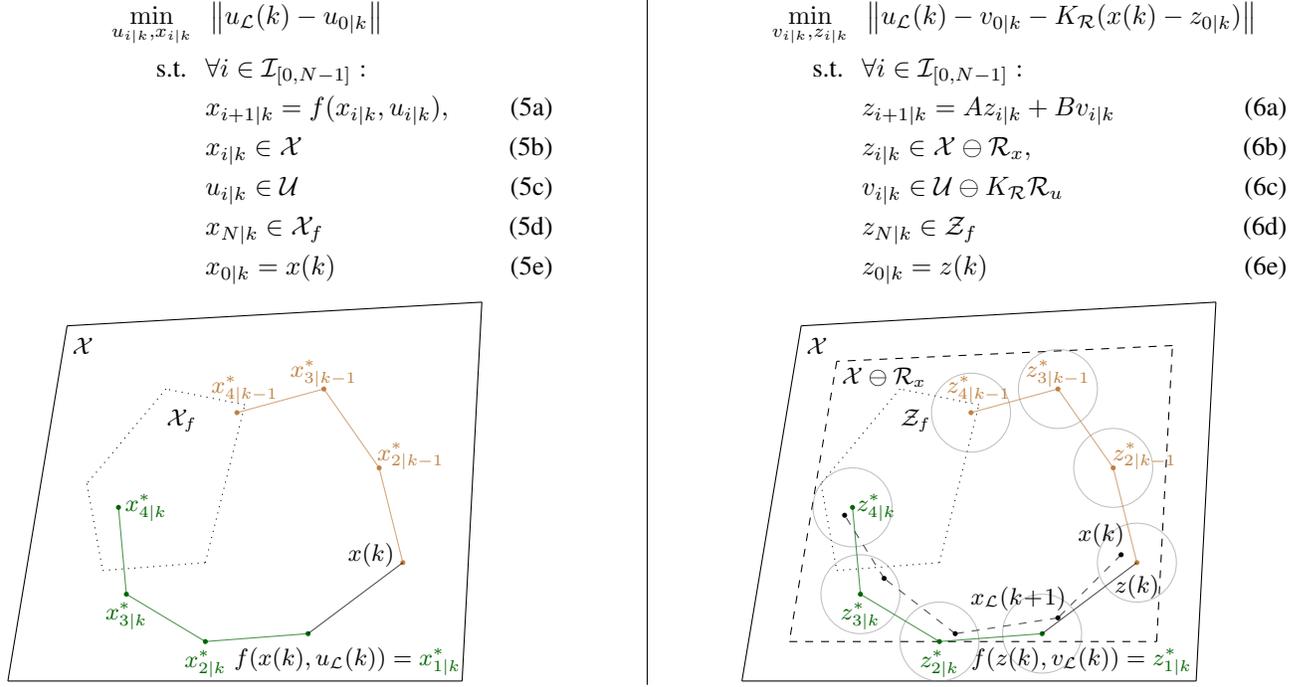
\begin{figure*}[ht]
\centering
\hfill
\begin{minipage}[b]{.4\textwidth}
      \begin{subequations}\label{eq:NMPSC_opt}
        \begin{align}
            \min_{\substack{\mUpred_{i|k},\mXpred_{i|k}}} ~ & \mNormGen{u_\LL(k) - \mUpred_{0|k}} \nonumber \\
            \text{s.t.} ~~ & \forall i\in\mIntInt{0}{N-1}: \nonumber \\
            & \mXpred_{i+1|k} = f(\mXpred_{i|k},\mUpred_{i|k}), \label{eq:NMPSC_opt_nominal_dynamics} \\
            & \mXpred_{i|k} \in \XX \\
            & \mUpred_{i|k} \in \UU \\
            & \mXpred_{N|k} \in \XX_f  \label{eq:NMPSC_opt_terminal}\\
            & \mXpred_{0|k} = x(k)  \label{eq:NMPSC_initial_nominal}
        \end{align}
      \end{subequations}
            \begin{tikzpicture}[scale = 1.05]
          \input{fig/tikz/pmpsc_scheme_nom.tex}
      \end{tikzpicture}
\end{minipage}
\unskip\ \hfill \vrule\ \hfill
\begin{minipage}[b]{.4\textwidth}
      \begin{subequations}\label{eq:PMPSC_opt}
        \begin{align}
            \min_{\substack{\mVpred_{i|k},\mZpred_{i|k}}} ~ & \mNormGen{u_\LL(k) - \mVpred_{0|k} - \mTubeControl(x(k) - \mZpred_{0|k})} \nonumber\\
            \text{s.t.} ~~ & \forall i\in\mIntInt{0}{N-1}: \nonumber \\
            & \mZpred_{i+1|k} = A\mZpred_{i|k} + B\mVpred_{i|k} \label{eq:PMPSC_opt_nominal_dynamics} \\
            & \mZpred_{i|k} \in \XX \ominus \mTube_x, \label{eq:PMPSC_opt_tightening_x} \\
            & \mVpred_{i|k} \in \UU \ominus \mTubeControl\mTube_u \label{eq:PMPSC_opt_tightening_u}\\
            & \mZpred_{N|k} \in \ZZ_f  \label{eq:PMPSC_opt_terminal}\\
            & \mZpred_{0|k} = z(k)  \label{eq:PMPSC_initial_nominal}
        \end{align}
      \end{subequations}
      \begin{tikzpicture}[scale = 1.05]
          \input{fig/tikz/pmpsc_scheme_tube.tex}
      \end{tikzpicture}
    \end{minipage}\hfill
    \caption{\small Mechanism in order to construct a safety policy `on-the-fly':
        The system is depicted at time $k$, with the current backup solution in brown.
        A proposed learning input $u_\LL$ is certified by constructing a safe solution
        for the following time step, shown in green.
        The existence of a safe trajectory is ensured by extending the
        brown trajectory using Assumptions~\ref{ass:nominal_terminal_invariant_set}
        and~\ref{ass:tightened_terminal_invariant_set} respectively.
        \textbf{Left (NMPSC):} Safe solutions are computed with respect to\ the
        true state dynamics, and constraints $x \in \XX$ are guaranteed
        to be satisfied.
        \textbf{Right (PMPSC):} Safe solutions are computed with respect to\ the
        nominal state $z$. The true state lies within the tube
        around the nominal state with probability $p_x$.
        By enforcing $z \in \XX \ominus \mTube_x$ constraint satisfaction holds with at
        least the same probability.
        }
    \label{fig:model_predictive_safety}    
\end{figure*}
The fundamental idea of \emph{model predictive safety certification},
which was introduced for linear deterministic systems
in~\cite{wabersich2018linear}, is the on-the-fly computation
of a safety policy $\pi_\mSafeSet$ that ensures
constraint satisfaction at all times in the future.
The safety policy is specified using MPC methods,
i.e., an input sequence is computed that safely steers the system to
a terminal safe set $\XX_f$, which can be done efficiently in
real-time. The first input is selected
as the learning input if possible, in which case it is certified as safe,
or selected as `close' as possible to the learning input otherwise. A specific
choice of the terminal safe set $\XX_f$ allows us to show that a previous
solution at time $k-1$ implies the existence of a feasible solution at time $k$,
ensuring safety for all future times.
Such terminal sets $\XX_f$ can, e.g., be a neighborhood of a locally stable
steady-state of the system~\eqref{eq:general_nonlinear_system}, or a possibly conservative
set of states for which a safe controller is known.

\subsection{Nominal model predictive safety certification scheme}\label{subsec:NMPSC}
In order to introduce the basic idea of the presented approach, we
introduce a \emph{nominal model predictive safety certification (NMPSC)} scheme
under the simplifying assumption that the system dynamics~\eqref{eq:general_nonlinear_system}
are perfectly known, time-independent, and without noise, i.e.\
$x(k+1) = f(x(k), u(k))~\forall k\in\mIntGeq{0}$.
The mechanism to construct the safety policy for certifying a given control input
is based on concepts from model predictive control
\cite{rawlings2009model} as illustrated in Figure~\ref{fig:model_predictive_safety}~(left)
with the difference that we aim at certifying an external learning signal $u_\LL$
instead of performing, e.g., safe steady-state or trajectory tracking. 

The safety policy $\pi_\mSafeSet$ is defined implicitly through
optimization problem~\eqref{eq:NMPSC_opt}, unifying certification and computation of
the safety policy. Thereby, problem~\eqref{eq:NMPSC_opt} does
not only describe the computation of a safety policy based on the current state $x(k)$
and according to Definition~\eqref{def:safety_certification}, but also provides a
mechanism in order to modify the learning-based control input $u_\LL(k)$ as little as necessary
in order to find a safe backup policy for the predicted state at time $k+1$.

In \eqref{eq:NMPSC_opt}, $x_{i|k}$ is the state predicted $i$ time steps ahead,
computed at time $k$, i.e.\ $x_{0|k} = x(k)$. Problem~\eqref{eq:NMPSC_opt} computes an
$N$-step input sequence $\{\mUpredOpt_{i|k}\}$ satisfying the input constraints
$\UU$, such that the predicted states satisfy the constraints $\XX$ and reach
the terminal safe set $\XX_f$ after $N$ steps, where $N \in \mIntGeq{1}$ is the
prediction horizon. The safety controller $\pi_\mSafeSet$ is defined as the first
input $\mUpredOpt_{0|k}$ of the computed optimal input sequence driving the system to 
the terminal set, which guarantees safety for all future times via an invariance property.

\begin{assumption}[Nominal invariant terminal set]\label{ass:nominal_terminal_invariant_set}
  There exists a nominal terminal invariant set $\XX_f\subseteq\XX$
  and a corresponding control law $\mDefFunction{\kappa_f}{\XX_f}{\UU}$, such that
  for all $x\in\XX_f$ it holds that $\kappa_f(x) \in \UU$ and
  $f(x,\kappa_f(x)) \in \XX_f$.
  \END
\end{assumption}

Assumption~\ref{ass:nominal_terminal_invariant_set} provides recursive feasibility of
optimization problem~\eqref{eq:NMPSC_opt} and therefore infinite-time constraint
satisfaction, i.e., if a feasible solution at time $k$ exists, one exists at
$k\!+\!1$ and therefore at all future times, see i.e.\ \cite{rawlings2009model}.

The safety certification scheme then works as follows. Consider a measured system state
$x(k-1)$, for which~\eqref{eq:NMPSC_opt} is feasible and the input trajectory
$\{\mUpredOpt_{i|k-1}\}$ is computed. After applying the first input to the system
$u(k-1)=\mUpredOpt_{0|k-1}$, the resulting state $x(k)$ is measured again. Because it
holds in the nominal case that $x(k) = \mXpredOpt_{1|k-1}$, a valid input sequence
$\{\mUpredOpt_{1|k-1},\ldots,\mUpredOpt_{N-1|k-1}, \kappa_f(\mXpredOpt_{N|k-1})\}$ is known
from the previous time step, which satisfies constraints and steers the state to
the safe terminal set $\XX_f$, as indicated by the brown trajectory in
Figure~\ref{fig:model_predictive_safety}~(left).
The safety of a proposed learning input $u_\LL$ is certified by solving optimization
problem~\eqref{eq:NMPSC_opt}, which if feasible for $\mUpred_{0|k} = u_\LL$,
provides the green trajectory in Figure~\ref{fig:model_predictive_safety}~(left)
such that $u_\LL$ can be safely applied to the system.
Should problem~\eqref{eq:NMPSC_opt} not be feasible for $\mUpred_{0|k} = u_\LL$,
it returns an alternative input sequence that safely guides the system towards
the safe set $\XX_f$. The first element of this sequence $\mUpredOpt_{0|k}$ is chosen to be as close
as possible to $u_\LL$ and is applied to system~\eqref{eq:general_nonlinear_system} instead
of $u_\LL$. Due to recursive feasibility, i.e., knowledge of the brown trajectory in
Figure~\ref{fig:model_predictive_safety} (left), such a
solution always exists, ensuring safety.

In the context of learning-based control, the true system dynamics
are rarely known accurately. In order to derive a probabilistic version
of the NMPSC scheme that accounts for uncertainty in the system
model~\eqref{eq:general_linear_additive_description} in the following,
we leverage advances in probabilistic stochastic model predictive control~\cite{hewing2018stochastic},
based on so-called probabilistic reachable sets.

\subsection{Probabilistic model predictive safety certification scheme}\label{subsec:PMPSC}
In the case of uncertain system dynamics, the safety policy consists of two components
following a tube-based MPC concept \cite{rawlings2009model}.
The first component considers a nominal state of the system $z(k)$ driven by
linear dynamics, and computes a nominal safe trajectory $\{\mZpredOpt_{i|k},
\mVpredOpt_{i|k}\}$ through optimization problem~\eqref{eq:PMPSC_opt}, which is similar to the case of perfectly known dynamics
introduced in the previous section, defining the nominal input $v(k) = v^*_{0|k}$. The second component consists of an auxiliary controller,
which acts on the deviation $e(k)$ of the true system state from the nominal one
and ensures that the true state $x(k)$ remains close to the nominal trajectory.
Specifically, it guarantees that $e(k)$ lies within a set $\mTube$, often
called the `tube', with probability of at least $p_x$ at each time step.
Together, the resulting safety policy is able to steer the system state $x(k)$ within
the probabilistic tube along the nominal trajectory towards the safe terminal set.

We first define the main components and assumptions, in order to then introduce
the probabilistic model predictive safety certification (PMPSC) problem together
with the proposed safety controller.
Define with $z(k)\in \RR^n$ and $v(k)\in\RR^m$ the nominal system
states and inputs, as well as the nominal dynamics according to
model \eqref{eq:general_linear_additive_description} as
\begin{align}\label{eq:nominal_system}
	z(k\!+\!1) = Az(k) + Bv(k), ~k\in\mIntGeq{0}
\end{align}
with the initial condition $z(0)=x(0)$. For example, one might choose matrices $(A,B)$
in the context of learning time-invariant linear systems based on
the maximum likelihood estimate of the true system dynamics.
Denote $e(k)\mDef x(k)-z(k)$ as the error (deviation) between the true system state,
evolving according to~\eqref{eq:general_nonlinear_system},
and the nominal system state following~\eqref{eq:nominal_system}.
The controller is then defined by augmenting
the nominal input with an auxiliary feedback on the error, in the case of a
linear system \eqref{eq:nominal_system} a linear
state feedback controller $\mTubeControl$
\begin{align}\label{eq:auxiliary_control}
	u(k) = v(k) + \mTubeControl(x(k)-z(k)),
\end{align}
which keeps the real system state $x(k)$ close to the nominal system state
$z(k)$, i.e.\ keeps the error $e(k)$ small, 
if $\mTubeControl\in\RR^{m\times n}$ is chosen such that it stabilizes
system \eqref{eq:nominal_system}. By Assumption~\ref{ass:bound_on_nonlinearities},
the model error $w_{\theta}( x(k), u(k))$ is contained in
$\mAddDisturbance_{\theta}$ for all time steps with probability $p_{\theta}$.
Therefore, we drop the state and input dependencies in the following and simply refer to $w_{\theta}(k)$
as the model mismatch at time $k$, such that the error dynamics can be expressed as
\begin{align}\label{eq:error_dynamics}
    e(k\!+\!1) &= x(k\!+\!1) - z(k\!+\!1) \nonumber \\
    &= f_{\theta}(x(k),u(k)) + w_s(k) - Az(k) - Bv(k) \nonumber \\
    &= f_{\theta}(x(k),u(k)) - Ax(k) - Bu(k) \nonumber \\ 
    &\quad + Ax(k) + Bu(k) + w_s(k) - Az(k) - Bv(k) \nonumber \\
    &= (A + BK_\mTube)e(k) + w_{\theta}(k) + w_s(k).
\end{align}

By setting the initial nominal state to the real state, i.e., $z(0) = x(0)\Rightarrow e(0) = 0$,
the goal is to keep the evolving error $e(k)$,
i.e.\ the deviation from the nominal reference trajectory, small in probability with levels
$p_x$ and $p_u$ for state and input constraints \eqref{eq:chance_constraints}, respectively.
This requirement can be formalized using the concept of probabilistic reachable
sets introduced in~\cite{Pola2006,Abate2008,hewing2018stochastic}.

\begin{definition}\label{def:PRS}
  A set $\mPRS$ is a probabilistic reachable set (PRS) at probability level $p$
  for system \eqref{eq:error_dynamics} if
  \begin{align}
    e(0) = 0 \Rightarrow \Pr(e(k)\in\mPRS)\geq p,
  \end{align}
  for all $k\in\mIntGeq{0}$.\END
\end{definition}
In Section \ref{subsec:model_design} we show how to compute PRS sets
$\mTube_x,\mTube_u$, corresponding to state and input chance 
constraints~\eqref{eq:chance_constraints}, in order to fulfill the following Assumption.
\begin{figure*}[t]
	\centering
	\vspace{0.45cm}
	\begin{tikzpicture}[scale = 1]
		\input{fig/tikz/pmpsc_proof.tex}
	\end{tikzpicture}
	\caption{\small Illustration of the idea underlying the proof of Theorem~\ref{thm:pmpsc}
		without stochastic noise, i.e.\ $w_s(k)=0$, and $\mAddDisturbance_{\theta}$ polytopic.
        Starting from $x(k)$, the set of possible reachable states for $x(k+1)$ under the safety policy
        $u(k) = \mVpredOpt_{1|k-1} + K_\Omega(x(k) - \mZpredOpt_{1|k-1})$ from the previous
        time step is indicated by the three dotted black arrows. The corresponding predicted error set with respect
        to the nominal system is given by $\{(A+BK_\mTube)e(k) \oplus \mAddDisturbance_{\theta}\}$ as shown in red.
        Solving \eqref{eq:PMPSC_opt} yields the optimal input $u(k) = \mVpredOpt_{0|k} + K_\Omega(x(k) - \mZpredOpt_{0|k})$,
        which preserves the predicted error set, enabling us to probabilistically bound the error within the PRS
        $\mTube_x,\mTube_u$ from Assumption~\ref{ass:general_PRS_tube}.}
	\label{fig:proof_pmpsc}
\end{figure*}
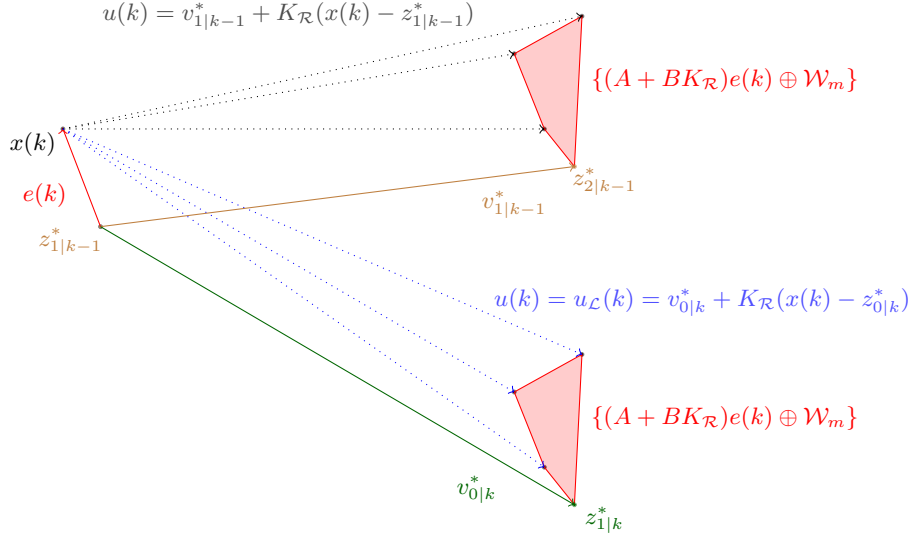
\begin{assumption}[Probabilistic tube]\label{ass:general_PRS_tube}
  There exists a linear state feedback matrix $\mTubeControl\in\RR^{m\times n}$ that stabilizes
  system~\eqref{eq:nominal_system}. The corresponding PRS sets for the error dynamics~\eqref{eq:error_dynamics}
  with probability levels $p_x$ and $p_u$ are denoted by $\mTube_x,\mTube_u\subseteq\RR^n$.
  \END
\end{assumption}
Based on Assumption~\ref{ass:general_PRS_tube}, it is possible to define
deterministic constraints on the nominal system \eqref{eq:nominal_system} that capture the chance 
constraints~\eqref{eq:chance_constraints}, by choosing $\XX\ominus\mTube_x$ 
as tightened state constraints, as depicted in
Figure~\ref{fig:model_predictive_safety}~(right), in which the grey circles
illustrate the PRS centered around the predicted nominal state $z$, such that
they contain the true state $x$ with probability $p_x$.
An appropriate tightening of the input constraints is obtained by linearly transforming
the error set $\mTube_u$ at probability level $p_u$ using the linear error feedback
$\mTubeControl$, resulting in $\UU \ominus \mTubeControl\mTube_u$.
%%%%% Insert for final submission:
Through calculation
of the nominal safety policy towards a safe terminal set within the tightened
constraints, and by application of~\eqref{eq:auxiliary_control}, finite-time
chance-constraint satisfaction over the planning horizon $k+N$
follows directly by Definition~\ref{def:PRS} and Assumption~\ref{ass:general_PRS_tube}.
In order to provide `infinite' horizon safety through recursive feasibility of
\eqref{eq:PMPSC_opt}, we require a terminal invariant set for the nominal system
state $\ZZ_f$ similar to Assumption~\ref{ass:nominal_terminal_invariant_set},
which is contained in the tightened constraints.
\begin{assumption}[Nominal terminal set]\label{ass:tightened_terminal_invariant_set}
  There exists a terminal invariant set $\ZZ_f\subseteq \XX \ominus \mTube_x$
  and a corresponding control law $\mDefFunction{\kappa_f}{\ZZ_f}{\UU \ominus \mTubeControl \mTube_u}$ such that
  for all $z\in\ZZ_f$ it holds that $\kappa_f(z) \in \UU \ominus \mTubeControl \mTube_u$ and
  $Az + B \kappa_f(z) \in \ZZ_f$.
  \END
\end{assumption}
The generation of the terminal set $\ZZ_f$ based on collected measurement data is 
discussed in Section~\ref{subsec:design_Z_f}, allowing for a successive improvement
of the overall PMPSC performance.
In classical tube-based and related stochastic MPC methods, the nominal
system~\eqref{eq:nominal_system} is re-initialized at each time step in order
to minimize the nominal objective, which causes a reset
of the corresponding error system. While this works in a robust setting, it
prohibits a direct probabilistic analysis using PRS according to Definition~\ref{def:PRS},
that only provides statements about the autonomous error system, starting from
time $k=0$ and evolving linearly for all future times. Consequently and in
contrast to classical formulations, we compute~\eqref{eq:nominal_system}
via~\eqref{eq:PMPSC_initial_nominal}, which leads to the error
dynamics~\eqref{eq:error_dynamics} despite online replanning of the
nominal trajectory at each time step compared also to Figure~\ref{fig:proof_pmpsc}
accompanying the proof of Theorem~\ref{thm:pmpsc}.

Building on the tube-based controller structure,
an input is certified as safe if it can be represented in
the form of \eqref{eq:auxiliary_control} by selecting
$\mVpredOpt_0$ accordingly. Otherwise an alternative input
is provided ensuring that $\Pr(e(k)\in\mTube_x)\geq p_x$,
$\Pr(e(k)\in\mTube_u)\geq p_u$ for all $k\in\mIntGeq{0}$.
Combining this mechanism with the assumptions from above yield the
main result of the paper.

\begin{theorem}\label{thm:pmpsc}
  Let Assumptions~\ref{ass:general_PRS_tube}
  and~\ref{ass:tightened_terminal_invariant_set}
  hold. If~\eqref{eq:PMPSC_opt} is feasible for $z(0) = x(0)$,
  then system~\eqref{eq:general_nonlinear_system} under the control
  law~\eqref{eq:auxiliary_control} with $v(k) = \mVpredOpt_{0|k}$ resulting from
  the PMPSC problem~\eqref{eq:PMPSC_opt} is safe for all $u_\LL(k)$ and for all times, i.e.,\
  the chance constraints~\eqref{eq:chance_constraints} are satisfied for all $k\geq 0$.
  \END
\end{theorem}

\begin{proof}
    We begin by investigating the error dynamics under \eqref{eq:auxiliary_control}.
    By \eqref{eq:PMPSC_initial_nominal} it follows that $e(k)$ evolves
    for all $k \in \mIntGeq{0}$, despite re-optimizing $\mVpred_{i|k},\mZpred_{i|k},$
    based on $u_\LL(k)$ according to \eqref{eq:PMPSC_opt} at every time step, see also
    Figure~\ref{fig:proof_pmpsc}. Therefore $\Pr(e(k) \in \mTube_x) \geq p_x$ and
    $\Pr(e(k) \in \mTube_u) \geq p_u$ for all $k\in\mIntGeq{0}$ by
    Assumption~\ref{ass:general_PRS_tube} and $z(0) = x(0)$.
    
    Next, Assumption~\ref{ass:tightened_terminal_invariant_set} provides recursive feasibility
    of optimization problem~\eqref{eq:PMPSC_opt},
    i.e.\ if a feasible solution at time $k$ exists, one will always exist
    at $k\!+\!1$, specifically $\{\mVpredOpt_{1|k},\ldots,\mVpredOpt_{N-1|k}, \kappa_f(\mZpredOpt_{N|k})\}$
    is a feasible solution,
    which implies feasibility of \eqref{eq:PMPSC_opt} for all $k\geq 0$ by induction.
	
    Finally, by recursive feasibility it follows that $z(k)\in\XX\ominus\mTube_x$ and
    $v(k)\in\UU\ominus \mTubeControl\mTube_u$ for all $k\in\mIntGeq{0}$, implying in combination with
    $\Pr(e(k) \in \mTube_x) \geq p_x$ and $\Pr(e(k) \in \mTube_u) \geq p_u$
    for all $k\in\mIntGeq{0}$ that $\Pr\{ x(k) = z(k) + e(k) \in \XX\}\geq p_x$ and
	$\Pr\{ u(k) = v(k) + \mTubeControl e(k) \in \UU\}\geq p_u$ for all $k\in\mIntGeq{0}$.
	
	We therefore prove that if~\eqref{eq:PMPSC_opt} is feasible for $z(0) = x(0)$,
    \eqref{eq:auxiliary_control} will always provide a control input such that constraints
    \eqref{eq:chance_constraints} are satisfied.
\end{proof}

\begin{remark}[Recursive feasibility despite unbounded disturbances]\label{rem:comparison_stmpc}
    Various recent stochastic model predictive control approaches, which
    consider chance constraints in the presence of unbounded additive noise, are
    also based on constraint tightening (see \cite{Farina2013,Farina2016,Paulson2017,hewing2018stochastic}).
    The online MPC problem in these formulations can become infeasible due to the unbounded
    noise, which is compensated for by employing a recovery mechanism.
    In contrast, the proposed formulation offers
    inherent recursive feasibility, even for unbounded disturbance realizations $w_s(k)$,
    by simulating the nominal system and therefore not optimizing over $\mZpred_0$.
    While optimization over $\mZpred_0$ usually enables state feedback in tube-based stochastic
    model predictive control methods, we incorporate state feedback
    through the cost function of \eqref{eq:PMPSC_opt}. A similar strategy can also be
    used for recursively feasible stochastic MPC schemes, as presented
    in~\cite{hewing2018recursively}.
    \END
\end{remark}

\subsection{Discussion: Safe RL with predictive safety certification}
    % View as safe system and potential problems
    While the application of RL algorithms together with
    the presented PMPSC scheme for safety can be viewed as
    learning to control an inherently safe system of the form
    \begin{align*}
        x(k+1) = f_{\theta}(x(k), \pi_\mSafeSet(k, x(k), \tilde u(k))) +  w_s(k), ~ \forall k\in\mIntGeq{0}
    \end{align*}
    with virtual inputs $\tilde u(k) \in\RR^m$, the underlying
    safety mechanism $\pi_\mSafeSet$ can introduce significant
    additional nonlinearities into the system to be optimized, which may
    cause slower learning convergence.
    % Outline of the following
    In the following, we therefore outline practical and theoretical
    aspects to retain learning convergence, when using RL in conjunction
    with the presented PMPSC scheme to ensure safety.

    % Simple method
    A simple mechanism to alleviate this limitation has been introduced in the context
    of general safety frameworks \cite{akametalu2014reachabilitySafety} and
    is based on adding a term of the form
    $(u(k) - u_\LL(k))^\top R_\mSafeSet(u(k) - u_\LL(k))$,
    $R_\mSafeSet\succ 0$ to the overall learning objective, which
    accounts for deviations between the proposed and the applied input to
    system \eqref{eq:general_nonlinear_system}, implying
    a larger cost for unsafe policies.
    % more sophisticated learning algorithms require continuity
    While this can encourage learning convergence in practice, many successful
    RL algorithms such as Trust-Region-Policy-Optimization (TRPO) \cite{Schulman2015},
    deterministic actor-critic policy search \cite{silver2014} or
    Bayesian Optimization \cite{martinez2007active,frohlich2020bayesian} that provide
    theoretical properties in terms of approximate gradient steps or cumulative regret,
    require continuity of the applied control policy.
    % What is continuity good for?
    More precisely, continuity of the applied control policy is often a prerequisite
    to adjust the learning-based policy via approximate gradient
    directions of the overall learning objective and is typically required to establish
    regularization properties of the corresponding value function.

    % How we can establish continuity
    Using results from explicit model predictive control, it
    is indeed possible to establish that the PMPSC problem and therefore
    the resulting safe learning-based policy $\pi_\mathcal{\mSafeSet}$ is continuous
    in $x(k)$ and $u_{\mathcal L}(k)$, see e.g.~\cite[Theorem 1.4]{borrelli2003constrained}.
    % How to compute gradient
    Furthermore, by using the results from \cite{agrawal2019differentiable},
    it is even possible to compute the partial derivatives of~\eqref{eq:PMPSC_opt}
    with respect to the learning input, i.e.,
    $\frac{\partial}{\partial u_\mathcal{L}}\pi_\mathcal{\mSafeSet}(k, x, u_\mathcal{L})$,
    allowing explicit consideration of the effects of safety-ensuring actions
    during closed-loop control for efficient policy parameter updates.
    % In conclusion, we can preserve nice properties of the policy
    In brief, due to the convexity of~\eqref{eq:PMPSC_opt}, continuity properties
    of the policy and value function are typically preserved, and it is possible
    to obtain comparable convergence properties of RL algorithms in combination
    with PMPSC, as illustrated in Figure~\ref{fig:performance} in 
    the numerical example section.

    % limitation of performance
    Note that even under these considerations the combination of a potentially unsafe
    RL algorithm with the PMPSC scheme presented can cause performance to deteriorate
    compared to the plain application of a potentially unsafe RL algorithm, which is,
    however, conceptually inevitable when restricting the space of possible inputs in
    favor of safety guarantees.
    
\section{Data based design}\label{sec:data_driven_design}
In order to employ the PMPSC scheme the following components must be provided: The
model description \eqref{eq:general_linear_additive_description} of the true
system \eqref{eq:general_nonlinear_system} (Section~\ref{subsec:model_design}), the
probabilistic error tubes $\mTube_x$ and $\mTube_u$
based on the model \eqref{eq:general_linear_additive_description} according to
Assumption~\ref{ass:general_PRS_tube} (Section~\ref{subsec:design_omega}),
and the nominal terminal set $\ZZ_f$, which
provides recursive feasibility according to
Assumption~\ref{ass:tightened_terminal_invariant_set} (Section~\ref{subsec:design_Z_f}).
In this section, we present efficient techniques for computing those components
that are tailored to the learning-based control setup by solely relying
on the available data collected from the system.

\subsection{Model design}\label{subsec:model_design}
For simplicity, we focus on systems described by linear Bayesian regression
\cite{rasumssen2006Gaussian,Bishop2006} of the form
\begin{align}\label{eq:linear_bayesian_regression}
    x(k+1) = \theta^\top \begin{psmallmatrix}x\\u \end{psmallmatrix} + w_s(k)
\end{align}
with an unknown parameter matrix $\theta\in\RR^{n\times n + m}$, which is inferred
from noisy measurements $\mData = \{(x_i,u_i),y_i \}_{i=1}^{N_\mData}$ with
$y_k = f_{\theta}(x(k),u(k)) + w_s(k)$, $w_s(k)\sim \mDistribution{w_s}$,
using a prior distribution $\mDistribution{\theta}$ on the parameters $\theta$.
Note that distribution pairs $\mDistribution{w_s}$ and $\mDistribution{\theta}$
that allow for efficient posterior computation, e.g.\ Gaussian distributions, usually exhibit
infinite support, i.e.\ $w_s(k)\in\RR^n$, which can generally not be treated robustly using, e.g.,
the related method presented in \cite{wabersich2018linear}. The correct
selection of the parameter prior and process noise distributions
$\mDistribution{\theta}$ and $\mDistribution{w_s}$ require careful model
selection techniques that are beyond the scope of this paper,
see, e.g., \cite[Section 2.3]{rasumssen2006Gaussian} and
\cite[Chapter 3]{Bishop2006}. In addition, we assume that the
measurements $\mData$ are sufficiently
information-rich, e.g., that they have been generated using
excitation signals as described in~\cite{ljung1998system}.

In the following, we present one way of obtaining the required model error set
$\mAddDisturbance_{\theta}$ using confidence sets based on
the posterior distribution $\mDistribution{\theta|\mData}$.

We start by describing the set of all
realizations of \eqref{eq:linear_bayesian_regression}, which contain the true
system with probability $p_{\theta}$. To this end let the confidence region at probability level $p_{\theta}$ of
the random vector $\theta\sim\mDistribution{\theta|\mData}$, denoted by
$\mConfidenceSet{\mDistribution{\theta|\mData}}{p_{\theta}}$, be defined such that
\begin{align}\label{eq:confidence_set}
    \Pr(\theta\in\mConfidenceSet{\mDistribution{\theta|\mData}}{p_{\theta}})\geq p_{\theta}
\end{align}
and compare the corresponding set of system dynamics with
the expected system dynamics, which is in the considered case given by
\begin{align}\label{eq:expected_linear_dynamics}
    \underbrace{\mExpectation{\theta}^\top}_{=:[A, B]} \begin{psmallmatrix}x(k)\\u(k) \end{psmallmatrix}.
\end{align}
Note that the model error between~\eqref{eq:linear_bayesian_regression} and
\eqref{eq:expected_linear_dynamics} is unbounded by definition, if we consider an unbounded domain as
required by \eqref{eq:bound_on_nonlinearities} since
$\lim_{\mNorm{(x,u)}{2}\rightarrow\infty}||((A,B)-\tilde \theta^\top)(x,u)||_2=\infty$
for any $\tilde \theta\in\mConfidenceSet{\mDistribution{\theta|\mData}}{p_{\theta}}$
such that $\tilde \theta \neq \theta$.
We therefore make the practical assumption that the model error is bounded outside a
sufficiently large `outer' state and input space $\XX_o\times\UU_o\supseteq \XX \times \UU$,
as illustrated in Figure~\ref{fig:linear_model_estimate}, which relates to
Assumption~\ref{ass:bound_on_nonlinearities} as follows.
\begin{assumption}[Bounded model error]\label{ass:outer_inner_state_space}
	The set 
	\begin{align*}
		\tilde \mAddDisturbance_{\theta}\mDef\{ w_m\in\RR^n | & 
		\forall (x,u,\theta)\in \XX_o\times\UU_o\times\mConfidenceSet{\mDistribution{\theta|\mData}}{p_{\theta}} \\
		                                               & Ax + Bu + w_m = \theta^\top \begin{psmallmatrix}x\\u \end{psmallmatrix}\}.
	\end{align*}
    is an overbound of $\mAddDisturbance_{\theta}$ according to Assumption~\ref{ass:bound_on_nonlinearities},
    i.e.\ $\mAddDisturbance_{\theta} \subseteq \tilde \mAddDisturbance_{\theta}$. \END
\end{assumption}

A simple but efficient computation scheme for overapproximating
$\tilde \mAddDisturbance_{\theta}$ using $\mConfidenceSet{\mDistribution{\theta|\mData}}{p_{\theta}}$
can be developed for the special case of a Gaussian prior distribution
$\mCol{i}(\theta)\sim\NN(0,\Sigma^\theta_i)$ and
Gaussian distributed process noise $w_s(k)\sim\NN(0,I_n\sigma_s^2)$.
We begin with the posterior distribution
$\mDistribution{\theta|\mData}$ of $\theta$
conditioned on data $\mData$, given by
\begin{align}\label{eq:bayesian_regression_posterior}
    p(\mCol{i}(\theta) | \mData) = \NN(\sigma_s^{-2}C_i^{-1}X\mCol{i}(y),C_i^{-1}),
\end{align}
where $\mRow{i}(X)=\mFeature(x_i,u_i)^\top$,
$\mRow{i}(y)=y_i^\top$, and $C_i=\sigma_s^{-2}XX^\top + (\Sigma_i^{\theta})^{-1}$,
see, e.g.~\cite{Friedman2001,rasumssen2006Gaussian}.

Using the posterior distribution $\mDistribution{\theta|\mData}$ 
according to \eqref{eq:bayesian_regression_posterior}
we compute a polytopic outer approximation of $\mConfidenceSet{\mDistribution{\theta|\mData}}{p_{\theta}}$
in a second step, which can be used in order to finally obtain an approximation of
$\tilde \mAddDisturbance_{\theta}$ and therefore $\mAddDisturbance_{\theta}$ since
$\mAddDisturbance\subseteq \tilde \mAddDisturbance_{\theta}$ by Assumption~\ref{ass:outer_inner_state_space}.
To this end, we consider the vectorized model parameters $\mVec(\theta)$ and their confidence set
$\mConfidenceSet{\mDistribution{\mVec(\theta)}}{p_{\theta}}=
    \{ \mVec(\theta) \in \RR^{n^2+mn} | \mVec(\theta)^\top C \mVec(\theta) \leq \chi^2_{n^2+mn}(p_{\theta}) \}$,
where $\chi^2_{n^2+mn}$ is the chi-squared distribution of degree $n^2+mn$ and 
\begin{align}\label{eq:helper_matrix_outer_hull}
    C := 
    \begin{pmatrix}
        C_1    & 0      & \cdots & 0      \\
        0      & C_2    & \cdots & 0      \\
        \vdots & \ddots &        & \vdots \\
        0      & 0      & \vdots & C_n
    \end{pmatrix}
\end{align}
is the posterior covariance according to \eqref{eq:bayesian_regression_posterior}. A computationally cheap
outer approximation of $\mConfidenceSet{\mDistribution{\mVec(\theta)}}{p_{\theta}}$ can be obtained by
picking its major axes $\{\tilde \theta_i\}_{i=1}^{n^2+mn}$ using singular value decomposition of $C$,
which provide the vertices of an inner polytopic approximation $\mConvexHull{\{\tilde \theta_i\}_{i=1}^{n^2+mn}}$
of $\mConfidenceSet{\mDistribution{\mVec(\theta)}}{p_{\theta}}$.
Scaling this inner approximation by $\sqrt{n^2+mn}$ \cite{boyd2004convex}
yields vertices of an outer polytopic approximation of $\mConfidenceSet{\mDistribution{\mVec(\theta)}}{p_{\theta}}$
given by the convex hull $\mConvexHull{\{\theta_i\}_{i=1}^{n^2+mn}}$ with $\theta_i = \sqrt{n^2+mn}\tilde \theta_i$.

Based on this outer approximation of $\mConfidenceSet{\mDistribution{\mVec(\theta)}}{p_{\theta}}$, it
is possible to compute a corresponding outer approximation of $\tilde \mAddDisturbance_{\theta}$
as follows. Due to the convexity of $\mConvexHull{\{\theta_i\}_{i=1}^{n^2+mn}}$, it is sufficient to impose
\begin{align}\label{eq:disturbance_set_based_on_confidence_set}
    \theta_i^\top \begin{psmallmatrix}x\\u\end{psmallmatrix} - (Ax + Bu)\in \tilde\mAddDisturbance_{\theta}
\end{align}
for all $i\in\mIntInt{1}{n^2+mn}$, $x\in\XX_o$, $u\in\UU_o$,
since by definition of $\mConfidenceSet{\mDistribution{\mVec(\theta)}}{p_{\theta}}$ 
and Assumption~\ref{ass:outer_inner_state_space} we have
with probability at least $p_{\theta}$ that $\lambda_i(x,u) \geq 0$
with $\sum_{i=1}^{n^2}\lambda_i(x,u) = 1$ exists such that 
\begin{align*}
        f_{\theta}(x,u) & =  \theta(x,u)^\top \begin{psmallmatrix}x\\u\end{psmallmatrix}                                         \\
                & = \sum_{i=1}^{n^2}\lambda_i(x,u)\theta_i^\top\begin{psmallmatrix}x\\u\end{psmallmatrix}                      \\
                &\in \sum_{i=1}^{n^2}\lambda_i(x,u)\left( Ax + Bu \oplus \tilde \mAddDisturbance_{\theta} \right) \\
                & \in \left\{ A x +  B u\right\} \oplus \tilde \mAddDisturbance_{\theta}.
\end{align*}
Therefore, \eqref{eq:disturbance_set_based_on_confidence_set} can be used in order to
construct an outer approximation $\tilde \mAddDisturbance_{\theta} = \{w\in\RR^n|~||w||_2 \leq w_{\mathrm{max}}\}$, where
\begin{align}\label{eq:W_bound_Linear_System}
    w_{\mathrm{max}} := \max_{i\in\mIntInt{1}{n^2+mn}}\left(\max_{x\in\XX_o,u\in\UU_o}
        \left\vert\left\vert(\theta_i^\top-(A~ B))
            \begin{pmatrix} x \\ u \end{pmatrix}\right\vert\right\vert _2\right)
\end{align}
with $(A~B):=\mExpectation{\theta}^\top$, $\{\theta_i\}_{i=1}^{n^2+mn} = \sqrt{2}\{\tilde \theta_i\}_{i=1}^{n^2+mn}$ and
$\tilde \theta_i$ the major axes of $\mConfidenceSet{\mDistribution{\mVec(\theta)}}{p_{\theta}}$.

\begin{figure}
	\centering
	\begin{tikzpicture}[scale = 0.9]
		\input{fig/tikz/linear_model.tex}
	\end{tikzpicture}
    \caption{Bounded error (dashed red lines) between the nominal model $Ax$ (blue line) in
        \eqref{eq:general_linear_additive_description} and a sample of the true dynamics $f_{\theta}(x)$ (gray line)
        beyond the outer bound $\partial \XX_o$ of the state space $\XX$
        according to Assumption~\ref{ass:outer_inner_state_space}.}
        \label{fig:linear_model_estimate}
\end{figure}
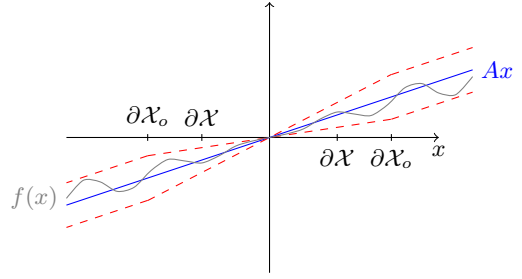

\subsection{Calculation of $\mTube$ for uncertain linear dynamics and unbounded disturbances}\label{subsec:design_omega}
In this subsection, we provide a method to compute a PRS set $\mTube$ 
with pre-specified probability level $p$ according to Assumption
\ref{ass:general_PRS_tube} that can be used for obtaining both, a PRS
$\mTube_x$ at probability level $p_x$ corresponding to the state constraints
\eqref{eq:chance_constraints}, and a PRS $\mTube_u$ at probability level $p_u$
for input constraints, respectively. As is common in the context of related MPC
methods, the computations are based on choosing a stabilizing tube controller
$K_\mTube$ in \eqref{eq:auxiliary_control} first, e.g.\ using LQR design,
in order to subsequently efficiently compute the PRS $\mTube_x$ and $\mTube_u$. 

The proposed PRS computation distinguishes between an error resulting from a
model mismatch between~\eqref{eq:general_nonlinear_system} and~\eqref{eq:general_linear_additive_description},
and an error caused by possibly unbounded process noise due to $w_s(k)\sim \mDistribution{w_s}$.
The error system~\eqref{eq:error_dynamics} admits a decomposition $e(k) = e_m(k) +
e_s(k)$ with $e(0) = e_m(0) = e_s(0) = 0$ and
\begin{align}
 e_m(k) &= (A+B\mTubeControl)e_m(k) + w_{\theta}(k), \label{eq:model_error_dynamics}\\
 e_s(k) &= (A+B\mTubeControl)e_s(k) + w_s(k). \label{eq:stochastic_error_dynamics}
\end{align}
While the fact that $\mDistribution{w_s}$ is known allows us to compute a
PRS with respect to $e_s(k)$ as described in \ref{subsec:design_omega} \textit{a)},
the model error $e_m(k)$ is state and input dependent with unknown distribution.
Therefore, we bound $e_m$ robustly in probability using the concept of robust invariance,
i.e., a robustly positive invariant set accounts deterministically for
all possible model mismatches $w_{\theta}(k)\in\mAddDisturbance_{\theta}$ at probability
level $p_{\theta}$ according to the following definition.
\begin{definition}\label{def:RIS}
    A set $\mathcal E$ is said to be a robustly positive invariant set (RIS) for system
    \eqref{eq:model_error_dynamics} if $e_m(0)\in\mathcal E$ implies
    that $e_m(k) \in\mathcal E$ for all $k\in\mIntGeq{0}$.
    It is called a RIS at probability level $p_{\theta}$ if
    $\Pr(\mathcal E \text{ is RIS} )\geq p_{\theta}$.
    \END
\end{definition}
This enables us to state the following lemma for the cumulated error $e(k)$
according to \eqref{eq:error_dynamics}.
\begin{lemma}\label{lem:split_PRS_sets}
    If $\mathcal E$ is a RIS at probability level $p_{\theta}$ for
    the model error system~\eqref{eq:model_error_dynamics} and $\mTube_s$
    is a PRS for the disturbance error system~\eqref{eq:stochastic_error_dynamics}
    at probability level $p_s$, then $\mTube = \mathcal E \oplus \mTube_s$
    is a PRS for the cumulated error system \eqref{eq:error_dynamics}
    at probability level $p_{\theta} p_s$.
    \END
\end{lemma}

\begin{proof}
    By the definition of the Minkowski sum, $e_m(k)\in\mathcal E$ and $e_s(k)\in\mTube_s$ implies $e(k)\in\mTube$ for any
     $k\in\mIntGeq{0}$. Choosing $e_s(0)=e_m(0)=0$
     yields for all $k\in\mIntGeq{0}$
    \begin{align*}
        \Pr(e(k)\in\mTube) &\geq \Pr(e_m(k)\in\mathcal E \land e_s(k)\in \mTube_s) \\
                           &= \Pr(e_m(k)\in\mathcal E)\Pr(e_s(k)\in\mTube_s)\\
                           &\geq p_{\theta}p_s,
    \end{align*}
    due to independence, which proves the result.
\end{proof}
Lemma \ref{lem:split_PRS_sets} allows computation of the PRS $\mTube_s$ that accounts for
the stochastic disturbances \eqref{eq:stochastic_error_dynamics} independently of the
RIS $\mathcal E$, dealing with the model uncertainty $\tilde \mAddDisturbance$.
In the following we present one option for determining $\mTube_s$
and refer to \cite{hewing2018stochastic} for further computation methods. Then
an optimization problem for the synthesis of $\mathcal E$ is given based on the
model obtained from Section \ref{subsec:model_design}.
\paragraph{PRS $\mTube_s$ for stochastic errors}\label{par:design_prs}
Using the variance $\mVariance{\mDistribution{w_s}}$ and the Chebyshev bound,
a popular way to compute $\mTube_s$ is given by solving the
Lyapunov equation $A_{cl}^\top \Sigma_\infty A_{cl} - \Sigma_\infty = - \mVariance{\mDistribution{w_s}}$
for $\Sigma_\infty$, which yields the PRS 
\begin{align}\label{eq:stochastic_ris}
    \mTube_s = \{ e_s \in\RR^n | e_s^\top \Sigma_\infty e_s \leq \tilde p\}
\end{align}
with probability level $p = 1 - n_x/\tilde p$, see e.g.~\cite{hewing2018stochastic}.
Furthermore, if $\mDistribution{w_s}$ is a normal distribution, $\mTube_s$
with $\tilde p = \chi^2_n(p)$ is a PRS of probability level $p$, where $\chi^2_n(p)$
is the quantile function of the $n$-dimensional chi-squared distribution.
\paragraph{RIS $\mathcal E$ at probability level $p_\theta$ for ellipsoidal model errors}
Given a bound on the model error according to Assumption~\ref{ass:bound_on_nonlinearities}
of the form $\mAddDisturbance_{\theta} = \{ w\in\RR^n ~|~ w^\top Q^{-1} w \leq 1\}$ with
$Q\in\mSetPosSymMat{n}$, e.g.\ $Q^{-1} := I_n w_{\mathrm{max}}^{-2}$ using
\eqref{eq:W_bound_Linear_System}, we can make use of methods from robust
control \cite{Boyd1994} in order to construct a possibly
small set $\mathcal E = \{e|e^\top P e \leq \alpha\}$ at probability
level $p_{\theta}$ by solving
\begin{subequations}\label{eq:linear_omega_opt}
	\begin{align}\label{eq:linear_omega_opt_objective}
		  & \max_{\alpha^{-1},\tau_{0},\tau_{1}} ~ \alpha^{-1}                   \\\nonumber
		  & \text{s.t.} :                                                        \\
		 & \begin{pmatrix}
			A_{cl}^\top P A_{cl} - \tau_{0} P & A_{cl}^\top P                    \\
			P A_{cl}                          & P - \tau_{1} Q^{-1}
		\end{pmatrix} \preceq 0 \label{eq:linear_omega_opt_lmi} \\
        & ~1 - \tau_{0} -  \bar{p} \tau_{1} \alpha^{-1} \geq 0 \, ,             \\
		  & ~\tau_{0}, \tau_{1} > 0 \, \label{eq:tau_conds},
	\end{align}
\end{subequations}
where $P\in\mSetPosSymMat{n}$ has to be pre-selected using, e.g., the infinite horizon LQR
cost $x(k)^\top P x(k)$, corresponding to the LQR feedback $u(k)=\mTubeControl x(k)$.
As pointed out in~\cite{hewing2018pis}, optimization
problem~\eqref{eq:linear_omega_opt} has a monotonicity property in the
bilinearity $\tau_{1}\alpha^{-1}$ such that it can be efficiently solved using a bisection
on the variable $\alpha^{-1}$. A more advanced design procedure, yielding less conservative
robust invariant sets, can be found, e.g., in \cite{Limon2008}.

In summary, based on the uncertainty of the system parameters
with respect to their true values inferred from data
and by solving \eqref{eq:linear_omega_opt}, we obtain a RIS $\mathcal E$ at
probability level $p_{\theta}$, which contains the model error~\eqref{eq:model_error_dynamics}.
Together with the PRS $\mTube_s$ from Section~\ref{subsec:design_omega} \textit{a)},
Lemma~\ref{lem:split_PRS_sets} provides the overall PRS for
the error system~\eqref{eq:error_dynamics}, which is given by
$\mTube = \mTube_s \oplus \mathcal E$ at probability level $p_s p_{\theta}$. Note
that the ratio between $p_{\theta}$ and $p_s$ can be freely chosen in order to
obtain overall tubes $\mTube_x$, $\mTube_u$ at probability levels
$p_x=p_s^{\mTube_x}p_{\theta}^{\mTube_x}$ and $p_u=p_s^{\mTube_u}p_{\theta}^{\mTube_u}$
according to the chance constraints~\eqref{eq:chance_constraints}.

\subsection{Iterative construction of the terminal safe set $\ZZ_f$}\label{subsec:design_Z_f}
While the terminal constraint~\eqref{eq:PMPSC_opt_terminal} in combination with
Assumption~\ref{ass:tightened_terminal_invariant_set} is key in order to provide a safe backup
control policy $\pi_\mSafeSet$ for all future times, it can restrict
the feasible set of~\eqref{eq:PMPSC_opt}.
The goal is therefore to provide a large terminal set $\ZZ_f$ yielding potentially
less conservative modifications of the proposed learning-based control input $u_\LL$
according to \eqref{eq:PMPSC_opt}.

This can be iteratively achieved by recycling previously calculated
solutions to~\eqref{eq:PMPSC_opt}, starting from a potentially
conservative initial terminal set $\ZZ_f$ according to
Assumption~\ref{ass:tightened_terminal_invariant_set}.
Such an initialization can be computed using standard invariant set methods for linear systems, see
e.g.~\cite{blanchini1999setInvariance} and references therein.
Note that the underlying idea of iteratively enlarging the terminal set
is related to the concepts presented, e.g., in~\cite{Brunner2013, Rosolia2017}.

Let the set of nominal predicted states obtained from successfully solved instances
of~\eqref{eq:PMPSC_opt} be denoted by
$\bm z^*(k) = \{ z^*_j(x(i)), i \in \mathcal I_{[1,k]}, j\in\mathcal I_{[0,N]} \}$.

\begin{proposition}\label{prop:enlargement_of_nominal_terminal_set}
    If Assumption~\ref{ass:tightened_terminal_invariant_set} holds
    for $\ZZ_f$ and~\eqref{eq:PMPSC_opt} is convex, then the set
    \begin{align}\label{eq:enlargement_of_nominal_terminal_set}
        \ZZ_f^k \mDef \mConvexHull{\bm z^*(k)} \cup \ZZ_f
    \end{align}
    satisfies Assumption~\ref{ass:tightened_terminal_invariant_set}.
    \END
\end{proposition}
\begin{proof}
    We proceed in a manner similar to the proof of \cite[Theorem IV.2]{wabersich2018linear}.
    Let $z\in \ZZ_f^k$ and note that if \eqref{eq:PMPSC_opt} is convex, then
    the feasible set is a convex set, see e.g.~\cite{boyd2004convex}, and therefore
    $\mConvexHull{\bm z^*(k)}$ is a subset of the feasible set. From here, together with
    the fact that the system dynamics are linear, it follows for $z\in\mathrm{co}(\bm z^*(k))$ that multipliers
    $\lambda_{ij} \geq 0$, $\Sigma_{i,j} \lambda_{ij} = 1$ exist such that we have
    $z = \Sigma_{i,j} \lambda_{ij} \mZpredOpt_j(x(i))$ with corresponding input
    $\Sigma_{i,j} \lambda_{ij} \mVpredOpt_j(x(i))$ that satisfies state and input
    constraints due to the convexity of these sets.
    We can therefore explicitly state
    \begin{align*}
        \kappa_f^{\bm z^*(k)}(z) =
        \begin{cases}
            \Sigma_{i,j} \lambda_{ij} \mVpredOpt_j(x(i)), \text{if}~z\in\mConvexHull{\bm z^*(k)}, \\
            \kappa_f(z), ~\text{else}
        \end{cases}
    \end{align*}
    as the required nominal terminal control law according to
    Assumption~\ref{ass:tightened_terminal_invariant_set} since
    $Az + B\kappa_f^{\bm z^*(k)}(z) \in \ZZ_f^k$ follows from convexity
    of $\mConvexHull{\bm z^*(k)}$ and invariance of $\ZZ_f$.
    Noting that $\mConvexHull{\bm z^*(k)}\subseteq \XX \ominus \mTube_x$ and
    $\mVpredOpt_{k|0}\subseteq \UU\ominus \mTubeControl\mTube_u$ by~
    \eqref{eq:PMPSC_opt_tightening_x}, \eqref{eq:PMPSC_opt_tightening_u}
    shows that for all $z\in\ZZ_f^k$ a control law $\bar \kappa_f$ exists
    according to Assumption~\ref{ass:tightened_terminal_invariant_set}, which completes the proof.
\end{proof}

\subsection{Overall MPSC design procedure}
    Given the methods presented in this section, the MPSC problem synthesis
    from data can be summarized as follows.
    \begin{samepage}
    \begin{labelledsteps}
        \item Compute a linear dynamical model of the form
            \eqref{eq:general_linear_additive_description}
            based on available measurements $\mData$ by estimating
            $\theta$ in \eqref{eq:linear_bayesian_regression} 
            using~\eqref{eq:bayesian_regression_posterior}.
        \item Compute a polytopic confidence set of $\theta$ using
            a singular value decomposition of~\eqref{eq:helper_matrix_outer_hull} to
            obtain the model uncertainty set according to~\eqref{eq:W_bound_Linear_System}.
        \item Compute the PRS $\mTube_{s,x}$, $\mTube_{s,u}$ corresponding to the
            additive stochastic uncertainty according to \eqref{eq:stochastic_ris}.
        \item Compute the RIS $\mathcal E_{x}$, $\mathcal E_{u}$ at the desired probability
            levels based on the model uncertainty by solving~\eqref{eq:linear_omega_opt}.
        \item Perform the state and input constraint tightening with respect to the
            Minkowski sums $\mTube_x = \mTube_{s,x} \oplus \mathcal E_{x}$ and 
            $K_\mTube \mTube_u=K_\mTube\{\mTube_{s,u} \oplus \mathcal E_{u}\}$ to obtain
            \eqref{eq:PMPSC_opt_tightening_x} and \eqref{eq:PMPSC_opt_tightening_u}.
        \item Initialize $\ZZ_f = \{ 0 \}$ (principled
            ways in order to calculate less restrictive $\ZZ_f$ can be found for example in 
            \cite{blanchini1999setInvariance}).
    \end{labelledsteps}        
    Using these ingredients, any potentially unsafe learning-based control
    law $u_\mathcal{L}(k)$ can be safeguarded by solving~\eqref{eq:PMPSC_opt}
    and applying~\eqref{eq:auxiliary_control} at every time step $k$. When
    the constraint tightening~\eqref{eq:PMPSC_opt_tightening_x},
    \eqref{eq:PMPSC_opt_tightening_u} or the nominal terminal set constraint
    \eqref{eq:PMPSC_opt_terminal} are overly conservative,
    it is possible to make use of the system trajectories during closed-loop operation
    to improve the performance. Collected state measurements can be
    used to reduce model uncertainty, allowing tighter bounds on $w_m$
    and recomputation of $\mTube_x,\mTube_u$ that enables greater
    exploration of the system in the future as demonstrated in
    Figure~\ref{fig:tightening}. In addition, nominal trajectories
    can be used to enlarge $\ZZ_f$ according
    to~\eqref{eq:enlargement_of_nominal_terminal_set}.
    \end{samepage}
\section{Numerical example: Safely Learning to control a car}\label{sec:example}
In this section, we apply the proposed PMPSC scheme in order to
safely learn how to drive a simulated autonomous car along
a desired trajectory without leaving a narrow road. For the car
simulation we consider the dynamics
\begin{alignat}{6}
    &\dot x = v\cos(\psi) \qquad &&  \dot \psi = (v/L) \tan (\delta)\tfrac{1}{1+ (v/v_\mathrm{CH})} \nonumber \\  
    &\dot y = v\sin(\psi) \qquad &&  \dot \delta = (1/T_\delta) (u_\delta - \delta) \nonumber \\
    &\dot v = a             \qquad &&  \dot a = (1/T_a)(u_a - a), \label{eq:car_dynamics}
\end{alignat}
with position $(x,y)$ in world coordinates, orientation $\psi$, velocity $v$, acceleration $a$,
and steering angle $\delta$, where the acceleration rate is modeled by a first-order lag with respect
to the desired acceleration (system input) $u_a$, and the angular velocity of the steering angle is also
modeled by a first-order lag with respect to the desired steering angle (system input) $u_\delta$.
The system is subject to the state and input constraints
$||\delta|| \leq 0.7~\mathrm{[rad]}$, $||v||\leq 19.8~\mathrm{[m~s^{-1}]}$,
$-6\leq a \leq 2~\mathrm{[m~s^{-2}]}$, $||u_\delta|| \leq 1.39~\mathrm{[rad]}$,
and $-6\leq u_a \leq 2~\mathrm{[m~s^{-2}]}$, for which the true car dynamics can
be approximately represented by \eqref{eq:car_dynamics}, see e.g.~\cite{kuwata2009real},
with parameters $T_\delta = 0.08~\mathrm{[s]}$, $T_a= 0.3~\mathrm{[s]}$,
$L = 2.9~\mathrm{[m]}$, and $v_{\mathrm{CH}}=20~\mathrm{[m~s^{-2}]}$.
The system is discretized with a sampling time of $0.1~\mathrm{[s]}$.

The learning task is to find a control law that tracks a periodic reference trajectory on a narrow road, which
translates in an additional safety constraint $||y||\leq 1$. The terminal set
according to Assumption \ref{ass:tightened_terminal_invariant_set} is defined as the road center
with angles $\psi=\delta=0$ and acceleration $a=0$, which is a safe set
for~\eqref{eq:car_dynamics} with $\kappa_f = 0$. The planning horizon is selected
to $N=30$ and the model \eqref{eq:general_linear_additive_description} as well as the
PRS set $\mTube_x = \mTube_u$ with probability level $98\%$ is 
computed based on a $30$ second state and input trajectory according to Section \ref{sec:data_driven_design},
see supplementary material for further details.
We use Bayesian Optimization as described in \cite{Neumann2019} for learning a linear
control law, implementing a policy search method that automatically trades off exploration
of the parameter space and exploitation of promising subsets and which does not provide inherent safety
guarantees. As the cost function for each episode, we penalize the deviation from the reference
trajectory quadratically, i.e.,
\begin{align*}
    \sum_{i=1}^{60} (x_{\mathrm{ref}}(i) - x(i))^\top Q (x_{\mathrm{ref}}(i) - x(i))+ u_\LL(i)Ru_\LL(i),
\end{align*}
where $Q = \mathrm{diag}((1 ~1.5 ~1 ~1 ~100 ~100)^\top)$ and $R = \mathrm{diag}((1~1)^\top)$.

While the resulting learning episodes without the PMPSC framework would leave the safety constraints, i.e.\ the road, in
a significant number of samples, as shown in Figure~\ref{fig:pmpsc_learning_episodes}, 
the safety framework enables safe learning in every episode.

In Figure~\ref{fig:pmpsc_safety_interaction}, two example learning episodes from Figure~\ref{fig:pmpsc_learning_episodes}
are shown, where the size of the input modification through the safety framework is
indicated with different circle radii along the trajectories. While in the first episode the safety framework intervenes with the learning-based policy in order to ensure safety, the algorithm safely begins to converge after 30 episodes with significantly
less safety interventions.

In Figure~\ref{fig:performance} we compare the performance of the learning-based
control policy when applied directly against the performance of the safety-enhanced policy using PMPSC
and observe that the safety-ensuring actions yield a slightly slower convergence and
slightly worse performance after learning convergence on average compared
to direct application of the unsafe algorithm.

\begin{figure}[t]
	\centering
	\includegraphics[width=\linewidth]{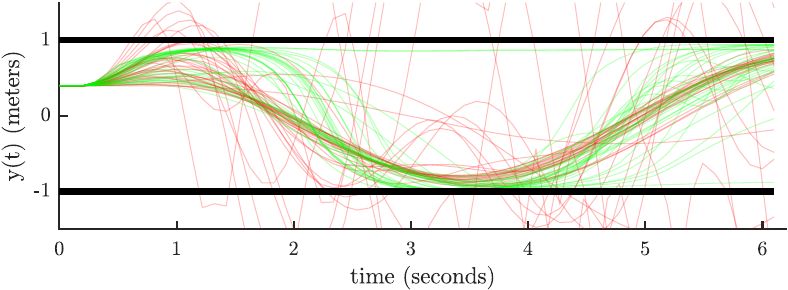}
    \caption{Learning to track a periodic trajectory subject to constraints:
        here we show the first $30$ learning episodes without (red) and with (green) the proposed PMPSC
        framework as well as the safety constraints (black).}\label{fig:pmpsc_learning_episodes}
\end{figure}
\begin{figure}[t]
    \centering
    \includegraphics[width=\linewidth]{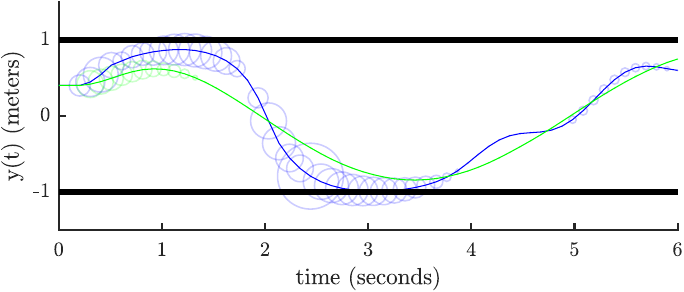}
    \caption{Resulting safe closed-loop trajectories during learning with
        initial policy parameters (blue) and final policy parameters (green).
        The circle radii indicate the relative magnitude of safety-ensuring modifications
        of the learning-based controller.}
    \label{fig:pmpsc_safety_interaction}
\end{figure}
\begin{figure}[t]
    \centering
    \includegraphics[width=\linewidth]{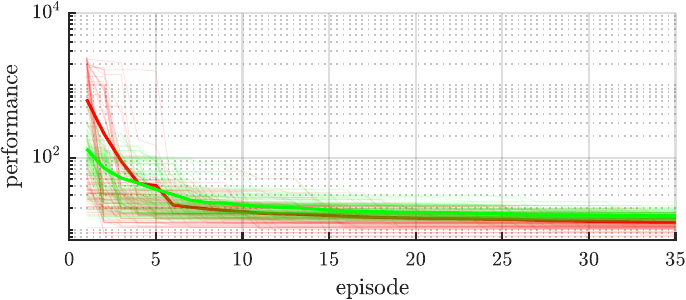}
    \caption{Closed-loop cost for 100 different experiments. Thin
    lines depict experiment samples and thick lines show the corresponding
    mean. Red lines indicate direct application of the learning-based
    controller and green lines illustrate the combination with the
    proposed PMPSC scheme.}\label{fig:performance}
\end{figure}

\section{Conclusion}
This paper has introduced a methodology to enhance arbitrary RL
algorithms with safety guarantees during the process of learning. The scheme is based on
a data-driven, linear belief approximation of the system dynamics 
that is used in order to compute safety policies for the
learning-based controller `on-the-fly'. By proving the existence of a safety policy
at all time steps, safety of the closed-loop system is established.
Principled design steps for the scheme are introduced, based on Bayesian
inference and convex optimization, which require little expert system knowledge
in order to realize safe RL applications.

\appendix
\subsection{Details of numerical example}
\begin{figure}[t]
    \centering
    \includegraphics[width=\linewidth]{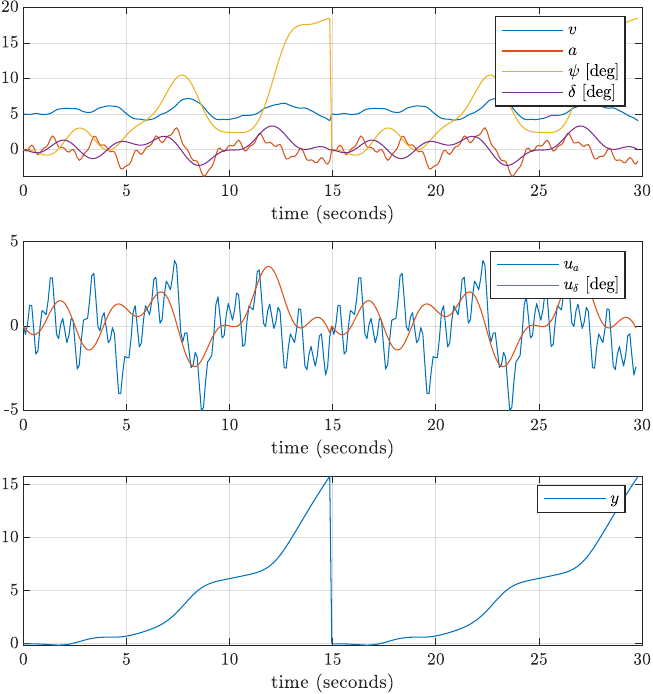}
    \caption{Measurements of system \eqref{eq:car_dynamics}, which are used
        for computing the PIS set $\mTube$.}\label{fig:measurements}
\end{figure}
\begin{figure}[t]
    \centering
    \includegraphics[width=\linewidth]{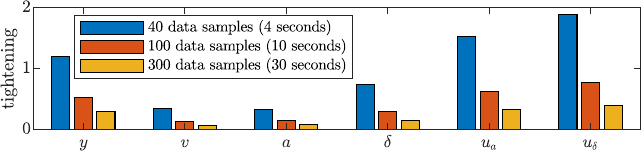}
    \caption{Tightening of state and input
        interval constraints for different numbers of measurements used
        to the design the PMPSC scheme.}\label{fig:tightening}
\end{figure}
\begin{figure}[t]
    \centering
    \includegraphics[width=\linewidth]{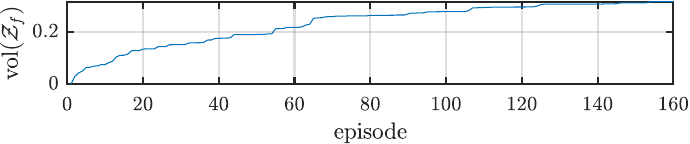}
    \caption{Volume of the terminal set $\ZZ_f$ according
    to~\eqref{eq:enlargement_of_nominal_terminal_set} over learning
    episodes.}\label{fig:z_enlargement}
\end{figure}

The model is computed according to Section \ref{subsec:model_design}
based on measurements of system \eqref{eq:car_dynamics} as
depicted in Figure \ref{fig:measurements}, sensor noise $\sigma_s = 0.01$ and prior distribution $\Sigma_i^p = 10I_n$. The state feedback
\begin{align*}
    \mTubeControl =
        \begin{pmatrix}
            -1.25 &  -0.05  & -2.34 &  -0.75 &  -0.19 &   0.02 \\
            0.02  & -0.69 &  -0.03 &   0.02 &  -5.04 &  -2.85
        \end{pmatrix}
\end{align*}
according to Assumption \ref{ass:general_PRS_tube} is computed according
to the mean dynamics of \eqref{eq:expected_linear_dynamics} using
LQR design.

Applying the procedure described in Section \ref{subsec:design_omega}
with different numbers of measurements as shown in
Figure~\ref{fig:measurements} yield different constraint tightenings
of the interval state and input constraints
as depicted in Figure~\ref{fig:tightening}. The final
tightened input and state constraints that are used
in the numerical example in Section~\ref{sec:example} are given as
$||u_\delta|| \leq 1.39~\mathrm{[rad]}$, $-5.4\leq u_a \leq 1.3~\mathrm{[m~s^{-2}]}$,
$||y|| \leq 0.87~\mathrm{[m]}$, $||\delta|| \leq 0.64~\mathrm{[rad]}$, $||v||\leq 19.97~\mathrm{[m~s^{-1}]}$,
$-4.97\leq a \leq 0.44~\mathrm{[m~s^{-2}]}$,
computed using the Yalmip-toolbox~\cite{Lofberg2004}
together with MOSEK \cite{mosek} to solve the resulting semi-definite program.

Starting from a terminal set $\ZZ_f=\{0\}$ we
illustrate in Figure~\ref{fig:z_enlargement} how
the volume of $\ZZ_f$ can iteratively be enlarged
based on previously calculated nominal
state trajectories at each time step by following Proposition~\ref{prop:enlargement_of_nominal_terminal_set} to reduce conservatism
of the terminal constraint~\eqref{eq:PMPSC_opt_terminal}.

\bibliographystyle{IEEEtran}
\bibliography{bibliography.bib}

% biography
\begin{IEEEbiography}[{\includegraphics[width=1in,height=1.25in,clip,keepaspectratio]{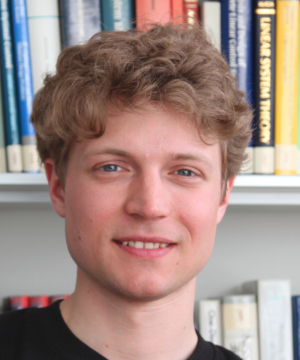}}]{Kim P. Wabersich}
    received his BSc. and MSc. in Engineering Cybernetics from the University of Stuttgart in Germany in 2015 and 2017, respectively.
    He is currently working towards a PhD. degree at the Institute for Dynamic Systems and Control (IDSC) at ETH Zurich.
    During his studies he was a research assistant at the Machine Learning and Robotics Lab (University of Stuttgart) and at the
    Daimler Autonomous Driving Research Center (B\"oblingen, Germany and Sunnyvale, CA).
    His research interests include safe learning-based control and model predictive control.
\end{IEEEbiography}
\begin{IEEEbiography}[{\includegraphics[width=1in,height=1.25in,clip,keepaspectratio]{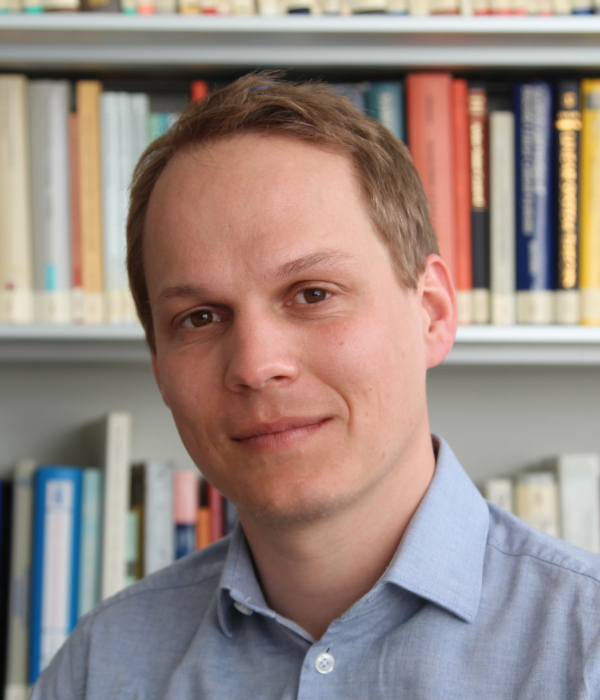}}]{Lukas Hewing}
    is a postdoctoral researcher at the Institute for Dynamic Systems and Control (IDSC) at ETH Zurich, from where he received his Ph.D. in 2020. Prior to this, he received his M.Sc. in Automation Engineering (with distinction) and B.Sc. in Mechanical Engineering from Aachen University in 2015 and 2013, respectively. He was a student research assistant a the Institute of Automatic Control (IRT) and Chair for Medical Information Technology (MedIT) in Aachen, Germany, and conducted a research stay at Tsinghua University, Beijing, China in 2015. His research interests include safe learning-based and stochastic model predictive control.
\end{IEEEbiography}
\begin{IEEEbiography}[{\includegraphics[width=1in,height=1.25in,clip,keepaspectratio]{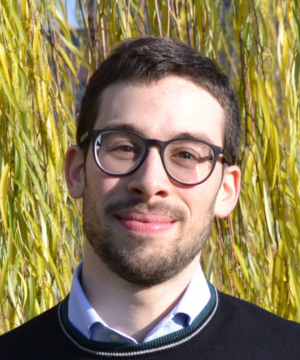}}]{Andrea Carron}
received the Dr. Eng. Bachelors and Masters 
degrees in control engineering from the University of Padova, Padova, Italy, in 2010 and 2012, respectively. 
He received his Ph.D. degree in 2016 from the University of Padova. 
He is currently a Postdoctoral Fellow with the Department of Mechanical and Process Engineering at ETH Zürich.
His interests include safe-learning, learning-based control, 
and nonparametric estimation. 
\end{IEEEbiography}
\begin{IEEEbiography}[{\includegraphics[width=1in,height=1.25in,clip,keepaspectratio]{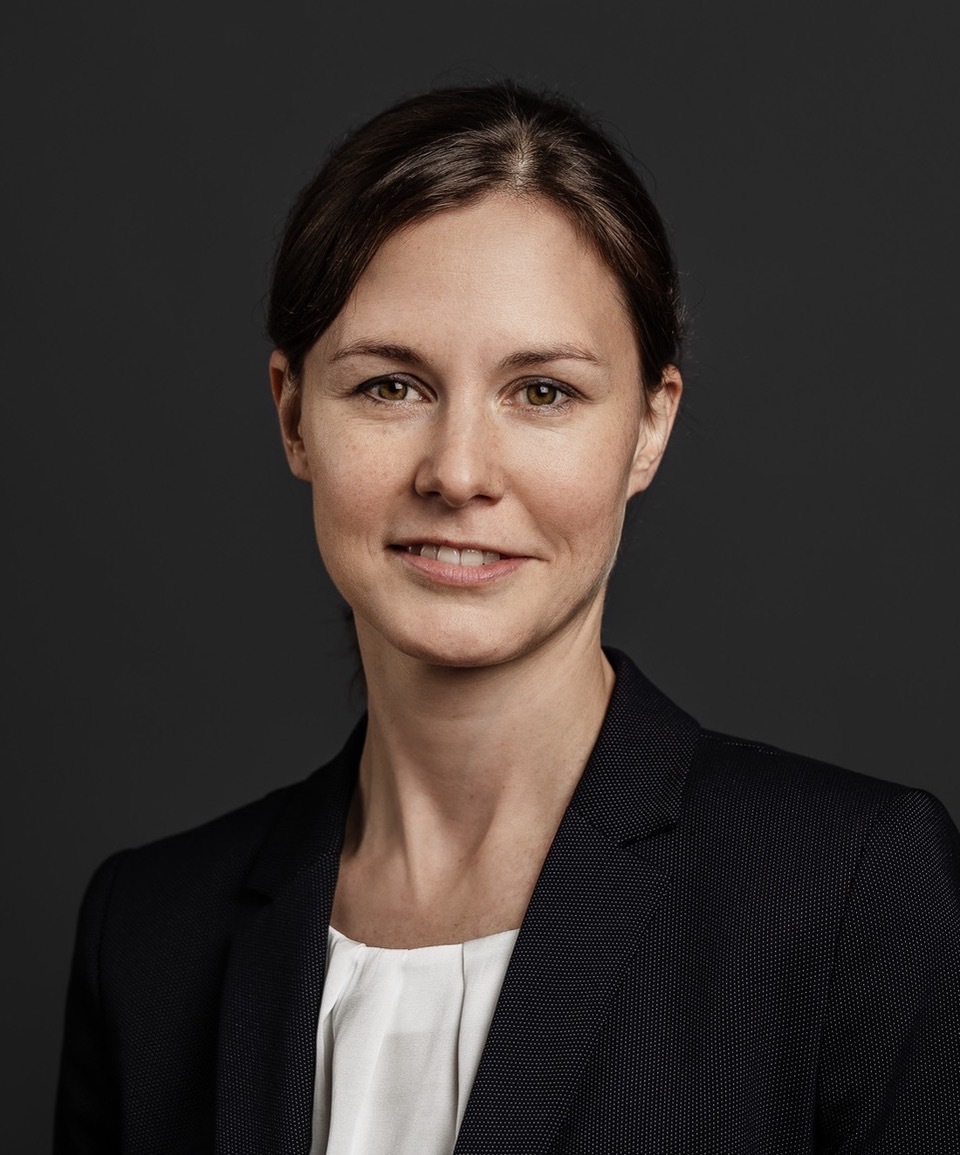}}]{Melanie N. Zeilinger}
    is an Assistant Professor at ETH Zurich, Switzerland. She received the Diploma degree in engineering cybernetics from the University of Stuttgart, Germany, in 2006, and the Ph.D. degree with honors in electrical engineering from ETH Zurich, Switzerland, in 2011. From 2011 to 2012 she was a Postdoctoral Fellow with the Ecole Polytechnique Federale de Lausanne (EPFL), Switzerland. She was a Marie Curie Fellow and Postdoctoral Researcher with the Max Planck Institute for Intelligent Systems, Tübingen, Germany until 2015 and with the Department of Electrical Engineering and Computer Sciences at the University of California at Berkeley, CA, USA, from 2012 to 2014. From 2018 to 2019 she was a professor at the University of Freiburg, Germany. Her current research interests include safe learning-based control, as well as distributed control and optimization, with applications to robotics and human-in-the-loop control.
\end{IEEEbiography}

\end{document}

%% file: definitions.tex
\theoremstyle{plain}
\newtheorem{theorem}{Theorem}[section]
\theoremstyle{plain}
\newtheorem{proposition}[theorem]{Proposition}
\theoremstyle{plain}
\theoremstyle{definition}
\newtheorem{definition}[theorem]{Definition}
\theoremstyle{plain}
\newtheorem{lemma}[theorem]{Lemma}
\theoremstyle{plain}

\theoremstyle{plain}
\newtheorem{remark}[theorem]{Remark}
\newtheorem{assumption}[theorem]{Assumption}
\theoremstyle{definition}

\newcommand{\RR}{\mathbb{R}}

\newcommand{\mSetSymMat}[1]{S^{#1}}
\newcommand{\mSetPosSemSymMat}[1]{S^{#1}_+}
\newcommand{\mSetPosSymMat}[1]{S^{#1}_{++}}
\newcommand{\mIntInt}[2]{\mathcal I_{[#1, #2]}}
\newcommand{\mIntGeq}[1]{\mathcal I_{\geq #1}}
\newcommand{\mNorm}[2]{\left\lVert {#1} \right\rVert_{#2}}
\newcommand{\mNormGen}[1]{\left\lVert {#1} \right\rVert}
\newcommand{\mDefFunction}[3]{#1: #2 \rightarrow #3}

\newcommand{\XX}{\mathcal{X}}
\newcommand{\UU}{\mathcal{U}}
\newcommand{\mSafeSet}{\mathcal{S}}

\newcommand{\mPRS}{\mathcal R}

\newcommand{\mAddDisturbance}{\mathcal W}
\newcommand{\mConvexHull}[1]{\mathrm{co}(#1)}

\newcommand{\mDef}{\coloneqq}
\newcommand{\mDistribution}[1]{\mathcal Q_{#1}}
\newcommand{\mConfidenceSet}[2]{\mathcal E_{#2}(#1)}
\newcommand{\mExpectation}[1]{\mathbb E(#1)}
\newcommand{\mVariance}[1]{\mathrm{var}(#1)}
\newcommand{\mFeature}{\phi}

\newcommand{\mData}{\mathcal D}
\newcommand{\mVec}{\mathrm {vec}}

\newcommand{\mTube}{\mathcal R}
\newcommand{\mTubeControl}{K_{\mathcal R}}
\newcommand{\mZpred}{{z}}
\newcommand{\mZpredOpt}{{z}^*}
\newcommand{\mVpred}{{v}}
\newcommand{\mVpredOpt}{{v}^*}

\newcommand{\mUpred}{u}
\newcommand{\mUpredOpt}{u^*}
\newcommand{\mXpred}{x}
\newcommand{\mXpredOpt}{x^*}

\newcommand{\NN}{\mathcal {N}}

\newcommand{\ZZ}{\mathcal{Z}}
\newcommand{\LL}{\mathcal{L}}

\newcommand{\mCol}[1]{\mathrm {col}_#1}
\newcommand{\mRow}[1]{\mathrm {row}_#1}

\newcommand*{\END}{\hfill\ensuremath{\lhd}}

\makeatletter
\newif\ifmygrid@coordinates
\tikzset{/mygrid/step line/.style={line width=0.80pt,draw=gray!80},
    /mygrid/steplet line/.style={line width=0.25pt,draw=gray!80}}
\pgfkeys{/mygrid/.cd,
    step/.store in=\mygrid@step,
    steplet/.store in=\mygrid@steplet,
    coordinates/.is if=mygrid@coordinates}
\def\mygrid@def@coordinates(#1,#2)(#3,#4){%
    \def\mygrid@xlo{#1}%
    \def\mygrid@xhi{#3}%
    \def\mygrid@ylo{#2}%
    \def\mygrid@yhi{#4}%
}
\newcommand\DrawGrid[3][]{%
    \pgfkeys{/mygrid/.cd,coordinates=true,step=1,steplet=0.2,#1}%
    \draw[/mygrid/steplet line] #2 grid[step=\mygrid@steplet] #3;
    \draw[/mygrid/step line] #2 grid[step=\mygrid@step] #3;
    \mygrid@def@coordinates#2#3%
    \ifmygrid@coordinates%
        \draw[/mygrid/step line]
        \foreach \xpos in {\mygrid@xlo,...,\mygrid@xhi} {%
                (\xpos,\mygrid@ylo) -- ++(0,-3pt)
                node[anchor=north] {$\xpos$}
            }
        \foreach \ypos in {\mygrid@ylo,...,\mygrid@yhi} {%
                (\mygrid@xlo,\ypos) -- ++(-3pt,0)
                node[anchor=east] {$\ypos$}
            };
    \fi%
}
\makeatother

%% file: fig/tikz/pmpsc_scheme_nom.tex
% !TeX root = ../../paper_NIPS/PMPSC.tex

\small

% help grid
%\DrawGrid{(-3,-3)}{(3,3)}

% Init coordinates for state space and tightened state space
\coordinate (XX1) at (-3,-2.5);
\coordinate (XX1T) at (-2.4,-2.0);
\coordinate (XX2) at (2.75,-2.5);
\coordinate (XX2T) at (2.25, -2.0);
\coordinate (XX3) at (3.0,2.3);
\coordinate (XX3T) at (2.45,1.75);
\coordinate (XX4) at (-2.25,2.0);
\coordinate (XX4T) at (-1.8,1.55);
% Draw state space
\draw (XX1) -- (XX2);
\draw (XX2) -- (XX3);
\draw (XX3) -- (XX4);
\draw (XX4) -- (XX1);
% Label state space
\node at (-2.05, 1.75) {$\XX$};

% Init coordinates for terminal region
\coordinate (XXf_offset) at (1,0);
\coordinate (XXf1) at (-2,1.2);
\coordinate (XXf2) at (-3,0);
\coordinate (XXf3) at (-2.8,-1.1);
\coordinate (XXf4) at (-1.5, -1);
\coordinate (XXf5) at (-1, 1);
% Draw terminal region
\draw[dotted] ($(XXf_offset) + (XXf1)$) -- ($(XXf_offset) + (XXf2)$);
\draw[dotted] ($(XXf_offset) + (XXf2)$) -- ($(XXf_offset) + (XXf3)$);
\draw[dotted] ($(XXf_offset) + (XXf3)$) -- ($(XXf_offset) + (XXf4)$);
\draw[dotted] ($(XXf_offset) + (XXf4)$) -- ($(XXf_offset) + (XXf5)$);
\draw[dotted] ($(XXf_offset) + (XXf5)$) -- ($(XXf_offset) + (XXf1)$);
\node at ($(-1.8, 0.8) + (XXf_offset)$) {$\XX_f$};

% Draw backup solution at time k-1
\coordinate (z1) at (2, -1);
\coordinate (z2) at (1.7,0.2);
\coordinate (z3) at (1, 1.2);
\coordinate (z4) at (-0.1, 0.9);
\coordinate (off) at (0, 0);

\draw[fill=brown, draw = brown] (z1) ellipse (0.025 and 0.025);
%\draw[thin, draw = gray, draw opacity = 0.5] (z1) ellipse (0.5 and 0.5);
\node at ($(z1) + (-0.4,0.1)$) {$x(k)$};
\draw[fill=brown, draw = brown] (z2) ellipse (0.025 and 0.025);
%\draw[thin, draw = gray, draw opacity = 0.5] (z2) ellipse (0.5 and 0.5);
\node[text = brown] at ($(z2) + (0.4,0.15)$) {$\mXpredOpt_{2|k-1}$};
\draw[fill=brown, draw = brown] (z3) ellipse (0.025 and 0.025);
%\draw[thin, draw = gray, draw opacity = 0.5] (z3) ellipse (0.5 and 0.5);
\node[text = brown] at ($(z3) + (0,0.2)$) {$\mXpredOpt_{3|k-1}$};
\draw[fill=brown, draw = brown] (z4) ellipse (0.025 and 0.025);
%\draw[thin, draw = gray, draw opacity = 0.5] (z4) ellipse (0.5 and 0.5);
\node[text = brown] at ($(z4) + (0.1,0.3)$) {$\mXpredOpt_{4|k-1}$};

\draw[thin,opacity=0.7, color = brown] (z1) -- (z2);
\draw[thin,opacity=0.7, color = brown] (z2) -- (z3);
\draw[thin,opacity=0.7, color = brown] (z3) -- (z4);
%\draw[->,thin,opacity=0.7, color = brown] (z4) -- (z5);

% Draw backup solution at time k
\coordinate (zz1) at (0.8, -1.9);
\coordinate (zz2) at (-0.5,-2);
\coordinate (zz3) at (-1.5,-1.4);
\coordinate (zz4) at (-1.6, -0.3);

\draw[fill=black!60!green, draw = black!60!green] (zz1) ellipse (0.025 and 0.025);
%\draw[thin, draw = gray, draw opacity = 0.5] (zz1) ellipse (0.5 and 0.5);
\node at ($(zz1) + (0.5,-0.35)$) {$ f(x(k),u_\LL(k)) = \textcolor{black!60!green}{\mXpredOpt_{1|k}}$};
\draw[fill=black!60!green, draw = black!60!green] (zz2) ellipse (0.025 and 0.025);
%\draw[thin, draw = gray, draw opacity = 0.5] (zz2) ellipse (0.5 and 0.5);
\node[text = black!60!green] at ($(zz2) + (0.0,-0.275)$) {$\mXpredOpt_{2|k}$};
\draw[fill=black!60!green, draw = black!60!green] (zz3) ellipse (0.025 and 0.025);
%\draw[thin, draw = gray, draw opacity = 0.5] (zz3) ellipse (0.5 and 0.5);
\node[text = black!60!green] at ($(zz3) + (0.0,-0.275)$) {$\mXpredOpt_{3|k}$};
\draw[fill=black!60!green, draw = black!60!green] (zz4) ellipse (0.025 and 0.025);
%\draw[thin, draw = gray, draw opacity = 0.5] (zz4) ellipse (0.5 and 0.5);
\node[text = black!60!green] at ($(zz4) + (0.35,0)$) {$\mXpredOpt_{4|k}$};

\draw[thin,opacity=0.7, color = black] (z1) -- (zz1);
\draw[thin,opacity=0.7, color = black!60!green] (zz1) -- (zz2);
\draw[thin,opacity=0.7, color = black!60!green] (zz2) -- (zz3);
\draw[thin,opacity=0.7, color = black!60!green] (zz3) -- (zz4);

%% file: fig/tikz/pmpsc_scheme_tube.tex
% !TeX root = ../../paper_NIPS/PMPSC.tex

\small

% help grid
%\DrawGrid{(-3,-3)}{(3,3)}

% Init coordinates for state space and tightened state space
\coordinate (XX1) at (-3,-2.5);
\coordinate (XX1T) at (-2.4,-2.0);
\coordinate (XX2) at (2.75,-2.5);
\coordinate (XX2T) at (2.25, -2.0);
\coordinate (XX3) at (3.0,2.3);
\coordinate (XX3T) at (2.45,1.75);
\coordinate (XX4) at (-2.25,2.0);
\coordinate (XX4T) at (-1.8,1.55);
% Draw state space
\draw (XX1) -- (XX2);
\draw (XX2) -- (XX3);
\draw (XX3) -- (XX4);
\draw (XX4) -- (XX1);
% Label state space
\node at (-2.05, 1.75) {$\XX$};

% Draw tightened state space
\draw[dashed] (XX1T) -- (XX2T);
\draw[dashed] (XX2T) -- (XX3T);
\draw[dashed] (XX3T) -- (XX4T);
\draw[dashed] (XX4T) -- (XX1T);
% Label tightened state space
\node at (-1.2, 1.35) {$\XX \ominus \mPRS_x$};

% Init coordinates for terminal region
\coordinate (XXf_offset) at (1,0);
\coordinate (XXf1) at (-2,1.2);
\coordinate (XXf2) at (-3,0);
\coordinate (XXf3) at (-2.8,-1.1);
\coordinate (XXf4) at (-1.5, -1);
\coordinate (XXf5) at (-1, 1);
% Draw terminal region
\draw[dotted] ($(XXf_offset) + (XXf1)$) -- ($(XXf_offset) + (XXf2)$);
\draw[dotted] ($(XXf_offset) + (XXf2)$) -- ($(XXf_offset) + (XXf3)$);
\draw[dotted] ($(XXf_offset) + (XXf3)$) -- ($(XXf_offset) + (XXf4)$);
\draw[dotted] ($(XXf_offset) + (XXf4)$) -- ($(XXf_offset) + (XXf5)$);
\draw[dotted] ($(XXf_offset) + (XXf5)$) -- ($(XXf_offset) + (XXf1)$);
\node at ($(-1.8, 0.8) + (XXf_offset)$) {$\ZZ_f$};

% Draw backup solution at time k-1
\coordinate (z1) at (2, -1);
\coordinate (z2) at (1.7,0.2);
\coordinate (z3) at (1, 1.2);
\coordinate (z4) at (-0.1, 0.9);
\coordinate (off) at (0, 0);

\draw[fill=brown, draw = brown] (z1) ellipse (0.025 and 0.025);
\draw[thin, draw = gray, draw opacity = 0.5] (z1) ellipse (0.5 and 0.5);
\node at ($(z1) + (0,-0.3)$) {$z(k)$};
\draw[fill=brown, draw = brown] (z2) ellipse (0.025 and 0.025);
\draw[thin, draw = gray, draw opacity = 0.5] (z2) ellipse (0.5 and 0.5);
\node[text = brown] at ($(z2) + (0.4,0.15)$) {$\mZpredOpt_{2|k-1}$};
\draw[fill=brown, draw = brown] (z3) ellipse (0.025 and 0.025);
\draw[thin, draw = gray, draw opacity = 0.5] (z3) ellipse (0.5 and 0.5);
\node[text = brown] at ($(z3) + (0,0.2)$) {$\mZpredOpt_{3|k-1}$};
\draw[fill=brown, draw = brown] (z4) ellipse (0.025 and 0.025);
\draw[thin, draw = gray, draw opacity = 0.5] (z4) ellipse (0.5 and 0.5);
\node[text = brown] at ($(z4) + (0.1,0.3)$) {$\mZpredOpt_{4|k-1}$};

\draw[thin,opacity=0.7, color = brown] (z1) -- (z2);
\draw[thin,opacity=0.7, color = brown] (z2) -- (z3);
\draw[thin,opacity=0.7, color = brown] (z3) -- (z4);
%\draw[->,thin,opacity=0.7, color = brown] (z4) -- (z5);

% Draw backup solution at time k
\coordinate (zz1) at (0.8, -1.9);
\coordinate (zz2) at (-0.5,-2);
\coordinate (zz3) at (-1.5,-1.4);
\coordinate (zz4) at (-1.6, -0.3);

\draw[fill=black!60!green, draw = black!60!green] (zz1) ellipse (0.025 and 0.025);
\draw[thin, draw = gray, draw opacity = 0.5] (zz1) ellipse (0.5 and 0.5);
\node at ($(zz1) + (0.5,-0.35)$) {$f(z(k),v_\LL(k)) = \textcolor{black!60!green}{\mZpredOpt_{1|k}}$};
\draw[fill=black!60!green, draw = black!60!green] (zz2) ellipse (0.025 and 0.025);
\draw[thin, draw = gray, draw opacity = 0.5] (zz2) ellipse (0.5 and 0.5);
\node[text = black!60!green] at ($(zz2) + (0.0,-0.275)$) {$\mZpredOpt_{2|k}$};
\draw[fill=black!60!green, draw = black!60!green] (zz3) ellipse (0.025 and 0.025);
\draw[thin, draw = gray, draw opacity = 0.5] (zz3) ellipse (0.5 and 0.5);
\node[text = black!60!green] at ($(zz3) + (0.0,-0.275)$) {$\mZpredOpt_{3|k}$};
\draw[fill=black!60!green, draw = black!60!green] (zz4) ellipse (0.025 and 0.025);
\draw[thin, draw = gray, draw opacity = 0.5] (zz4) ellipse (0.5 and 0.5);
\node[text = black!60!green] at ($(zz4) + (0.3,0)$) {$\mZpredOpt_{4|k}$};

\draw[thin,opacity=0.7, color = black] (z1) -- (zz1);
\draw[thin,opacity=0.7, color = black!60!green] (zz1) -- (zz2);
\draw[thin,opacity=0.7, color = black!60!green] (zz2) -- (zz3);
\draw[thin,opacity=0.7, color = black!60!green] (zz3) -- (zz4);

% Draw current system state
\coordinate (xk) at (1.8, -0.9);
\draw[fill=black] (xk) ellipse (0.025 and 0.025);
\node at ($(xk) + (off) + (-0.25,0.25)$) {$x(k)$};
% Draw actual system trajectory
\coordinate (x1) at (1, -1.7);
\coordinate (x2) at (-0.3, -1.9);
\coordinate (x3) at (-1.2, -1.2);
\coordinate (x4) at (-1.7, -0.4);
%\coordinate (x6) at (-4.2, -0.4);
\draw[fill=black] (x1) ellipse (0.025 and 0.025);
\node at ($(x1) + (off) + (-0.5,0.25)$) {$x_\LL(k\!+\!1)$};
\draw[fill=black] (x2) ellipse (0.025 and 0.025);
% \node at ($(x2) + (off) + (-0.1,-0.2)$) {$x(k+2)$};
\draw[fill=black] (x3) ellipse (0.025 and 0.025);
% \node at ($(x3) + (off) + (-0.5,-0.1)$) {$x(k+3)$};
\draw[fill=black] (x4) ellipse (0.025 and 0.025);
% \node at ($(x4) + (off) + (-1.05,0.1)$) {$x(k+4)$};
%\node[fill=white] at ($(x5) + (off) + (-0.9,-0.15)$) {$\mXpredOpt_{5|k}$};
%\node at ($(x4) + (off) + (-0.8,-0.15)$) {$x(k+5)$};
%\draw[fill=black] (x6) ellipse (0.025 and 0.025);

\draw[thin,opacity=0.7, dashed] (xk) -- (x1);
\draw[thin,opacity=0.7, dashed] (x1) -- (x2);
\draw[thin,opacity=0.7, dashed] (x2) -- (x3);
\draw[thin,opacity=0.7, dashed] (x3) -- (x4);
%\draw[->,thin,opacity=0.7] (x5) -- (x6);

%% file: fig/tikz/pmpsc_proof.tex
\small

% help grid
% \DrawGrid{(-1,3)}{(9,-7)}

\coordinate (coordinate_label_offset) at (-0.4, -0.2);
\coordinate (coordinate_label_offset2) at (0.4, -0.2);

% Init coordinates
\coordinate (x_k) at (0, 0);
\coordinate (e_k) at (0.16, -0.68);
\coordinate (z_1_km1) at (0.5, -1.3);

% x(k)
\node at ($(x_k) + (coordinate_label_offset)$) {$x(k)$};
\draw[ fill=black, fill opacity = 1, draw opacity = 0.5] (x_k) ellipse (0.025 and 0.025);

% e(k)
\node[text=red] at ($(e_k) + (coordinate_label_offset)$) {$e(k)$};

% z_{1|k-1}
\node[text=brown] at ($(z_1_km1) + (coordinate_label_offset)$) {$\mZpredOpt_{1|k-1}$};
\draw[ fill=brown, fill opacity = 1, draw opacity = 0.5] (z_1_km1) ellipse (0.025 and 0.025);

% error arrow
\draw[<-, thin,red] plot (x_k) -- (z_1_km1);

% error set init
\coordinate (a1_x) at (0.4,-1);
\coordinate (a2_x) at (0, 0);
\coordinate (a3_x) at (0.9,0.5);
\coordinate (z_2_km1) at (0.8, -1.5);
\coordinate (x_kp1) at (0.4, -0.25);

% error set 1
\coordinate (error_set_1) at (6, 1);
\draw[ draw = black, fill=black, fill opacity = 1, draw opacity = 0.5] ($(a1_x) + (error_set_1)$) ellipse (0.025 and 0.025);
\draw[ draw = black, fill=black, fill opacity = 1, draw opacity = 0.5] ($(a2_x) + (error_set_1)$) ellipse (0.025 and 0.025);
\draw[ draw = black, fill=black, fill opacity = 1, draw opacity = 0.5] ($(a3_x) + (error_set_1)$) ellipse (0.025 and 0.025);
\begin{scope}
    \draw[ draw = red , fill=red, fill opacity = 0.2] ($(a1_x) + (error_set_1)$) -- 
    ($(a1_x) + (error_set_1)$) --
    ($(a2_x) + (error_set_1)$) --
    ($(a3_x) + (error_set_1)$) --
    ($(z_2_km1) + (error_set_1)$) --
    ($(a1_x) + (error_set_1)$);
\end{scope}
% input candidate
\node[text=black!70] at (3,1.5) {$u(k) = \mVpredOpt_{1|k-1} + K_\mTube(x(k) - \mZpredOpt_{1|k-1})$};
% z_{2|k-1}
\node[text=brown] at ($(z_2_km1) + (coordinate_label_offset2) + (error_set_1)$) {$\mZpredOpt_{2|k-1}$};
\draw[fill=brown, fill opacity = 1, draw = brown, draw opacity = 0.5] ($(z_2_km1) + (error_set_1)$) ellipse (0.025 and 0.025);
% v_{1|k-1}
\node[text=brown] at (6,-1) {$\mVpredOpt_{1|k-1}$};
% bar x(k+1)
\node[text=red] at ($(x_kp1) + (coordinate_label_offset2) + (error_set_1) + (2,0.1)$) {$\{(A+BK_\mTube)e(k) \oplus \mAddDisturbance_m\}$};
% arrows
\draw[->, thin,brown] plot (z_1_km1) -- ($(z_2_km1) + (error_set_1)$);
\draw[->, thin,black,dotted] plot (x_k) -- ($(a1_x) + (error_set_1)$);
\draw[->, thin,black,dotted] plot (x_k) -- ($(a2_x) + (error_set_1)$);
\draw[->, thin,black,dotted] plot (x_k) -- ($(a3_x) + (error_set_1)$);

% error set 2
\coordinate (error_set_2) at (6, -3.5);
\draw[ draw = black, fill=black, fill opacity = 1, draw opacity = 0.5] ($(a1_x) + (error_set_2)$) ellipse (0.025 and 0.025);
\draw[ draw = black, fill=black, fill opacity = 1, draw opacity = 0.5] ($(a2_x) + (error_set_2)$) ellipse (0.025 and 0.025);
\draw[ draw = black, fill=black, fill opacity = 1, draw opacity = 0.5] ($(a3_x) + (error_set_2)$) ellipse (0.025 and 0.025);
\begin{scope}
    \draw[ draw = red , fill=red, fill opacity = 0.2] ($(a1_x) + (error_set_2)$) -- 
    ($(a1_x) + (error_set_2)$) --
    ($(a2_x) + (error_set_2)$) --
    ($(a3_x) + (error_set_2)$) --
    ($(z_2_km1) + (error_set_2)$) --
    ($(a1_x) + (error_set_2)$);
\end{scope}
% input candidate
\node[text=blue!70] at (8.5,-2.3) {$u(k) = u_\LL(k) = \mVpredOpt_{0|k} + K_\mTube(x(k) - \mZpredOpt_{0|k})$};
% z_{1|k}
\node[text=black!60!green] at ($(z_2_km1) + (coordinate_label_offset2) + (error_set_2)$) {$\mZpredOpt_{1|k}$};
\draw[fill=black!60!green, fill opacity = 1, draw = black!60!green, draw opacity = 0.5] ($(z_2_km1) + (error_set_2)$) ellipse (0.025 and 0.025);
% v_{0|k}
\node[text=black!60!green] at (5.5,-4.8) {$\mVpredOpt_{0|k}$};
% bar x(k+1)
\node[text=red] at ($(x_kp1) + (coordinate_label_offset2) + (error_set_2) + (2,0.1)$) {$\{(A+BK_\mTube)e(k) \oplus \mAddDisturbance_m\}$};
% arrows
\draw[->, thin,black!60!green] plot (z_1_km1) -- ($(z_2_km1) + (error_set_2)$);
\draw[->, thin,blue,dotted] plot (x_k) -- ($(a1_x) + (error_set_2)$);
\draw[->, thin,blue,dotted] plot (x_k) -- ($(a2_x) + (error_set_2)$);
\draw[->, thin,blue,dotted] plot (x_k) -- ($(a3_x) + (error_set_2)$);

%% file: fig/tikz/linear_model.tex
% !TeX root = ../../paper_NIPS/PMPSC.tex

\small

% help grid
% \DrawGrid{(-3,-3)}{(3,3)}

% Init coordinates for Axis
\coordinate (CX0) at (-3,0);
\coordinate (CX1) at (2.5,0);
\coordinate (CY0) at (0,-2);
\coordinate (CY1) at (0,2);

% Axis
\draw[->] (CX0) -- (CX1) node[below] {$x$};
\draw[->] (CY0) -- (CY1);

% Init state space
\coordinate (X0) at (-1,-0.05);
\coordinate (X1) at (1,-0.05);
\coordinate (Xo0) at (-1.8,-0.05);
\coordinate (Xo1) at (1.8,-0.05);

% draw state space
\coordinate (tick) at (0,0.1); % tick length
\draw (X0) -- ($(X0) + (tick)$) node[above] {$\partial \XX$};
\draw ($(X1) + (tick)$) -- (X1) node[below] {$\partial \XX$};
\draw (Xo0) -- ($(Xo0) + (tick)$) node[above] {$\partial \XX_o$};
\draw ($(Xo1) + (tick)$) -- (Xo1) node[below] {$\partial \XX_o$};

% Draw linear mean
\draw[scale=1,domain=-3:3,smooth,variable=\x,blue]
    plot ({\x},{(1/3)*\x})
    node[right] {$Ax$};

% Draw the linear cone
\draw[dashed,scale=1,domain=-1.8:1.8,smooth,variable=\x,red]
    plot ({\x},{0.518*\x});
\draw[dashed,scale=1,domain=-1.8:1.8,smooth,variable=\x,red]
    plot ({\x},{(0.15)*\x});

% Draw extending tube
\draw[dashed,scale=1,domain=1.8:3,smooth,variable=\x,red]
    plot ({\x},{0.33 + (1/3)*\x});
    % node[above right] {$Ax + \beta_{N_\mData}\sigma(x)$};
\draw[dashed,scale=1,domain=1.8:3,smooth,variable=\x,red]
    plot ({\x},{-0.33 + (1/3)*\x});
    % node[below right] {$Ax - \beta_{N_\mData}\sigma(x)$};

\draw[dashed,scale=1,domain=-1.8:-3,smooth,variable=\x,red]
    plot ({\x},{0.33 + (1/3)*\x});
\draw[dashed,scale=1,domain=-1.8:-3,smooth,variable=\x,red]
    plot ({\x},{-0.33 + (1/3)*\x});

% Draw the true nonlinear dynamics function
\draw[scale=1,domain=3:-3,smooth,variable=\x,gray]
    plot ({\x},{-0.4*sin(300*\x)*(1-exp(-0.25*abs(\x))) + 0.3 * \x })
    node[left] {$f(x)$};